\documentclass[sigconf]{acmart}
\AtBeginDocument{%
  }

\setcopyright{acmlicensed}
\copyrightyear{2025}
\acmYear{2025}
\acmDOI{XXXXXXX.XXXXXXX}
\acmConference[PPDP '25]{PPDP '25}{September 10–11, 2025}{Rende, Italy}
\acmBooktitle{Proceedings of the 27th International Symposium on Principles and Practice of Declarative Programming}
\acmISBN{979-8-4007-2085-7}

\usepackage{thm-restate}

\usepackage{latexsym}
\usepackage{ifthen}
\usepackage{mathtools}
\usepackage{cleveref}
\usepackage[all]{xy}

\usepackage{tikz}
\usetikzlibrary{arrows, decorations.pathmorphing}

\AtEndPreamble{%
  \theoremstyle{acmdefinition}
  \newtheorem{remark}[theorem]{Remark}
}

\newcommand{\m}[1]{\mathsf{#1}}
\newcommand{\mr}[1]{\mathrel{#1}}
\newcommand{\seq}[2][n]{{#2_1},\dots,{#2_{#1}}}
\newcommand{\sig}[2][n]{{#2_1}\times\cdots\times{#2_{#1}}}
\newcommand{\h}[1][.3]{\hspace{#1mm}}
\newcommand{\SET}[1]{\{\h#1\h\}}
\newcommand{\inter}[1]{{\llbracket #1 \rrbracket}}

\newcommand{\xB}{\mathcal{B}}

\newcommand{\xE}{\mathcal{E}}
\newcommand{\xF}{\mathcal{F}}

\newcommand{\xI}{\mathcal{I}}
\newcommand{\xJ}{\mathcal{J}}
\newcommand{\xM}{\mathcal{M}}
\newcommand{\xR}{\mathcal{R}}
\newcommand{\xS}{\mathcal{S}}
\newcommand{\xT}{\mathcal{T}}

\newcommand{\xV}{\mathcal{V}}
\newcommand{\xFTe}{\xF_{\m{te}}}
\newcommand{\xFTh}{\xF_{\m{th}}}
\newcommand{\xSTe}{\xS_{\m{te}}}
\newcommand{\xSTh}{\xS_{\m{th}}}

\newcommand{\SigmaTh}{\Sigma_{\m{th}}}

\newcommand{\Val}{\xV\m{al}}
\newcommand{\Var}{\xV\m{ar}}
\newcommand{\FVar}{\xF\xV\m{ar}}
\newcommand{\BVar}{\xB\xV\m{ar}}

\newcommand{\LVar}{\mathcal{L}\Var}

\newcommand{\ExVar}{\mathcal{E}\m{x}\Var}
\newcommand{\Dom}{\mathcal{D}\m{om}}

\newcommand{\VDom}{\mathcal{V}\Dom}
\newcommand{\Pos}{\mathcal{P}\m{os}}

\newcommand{\FPos}{\Pos_\xF}

\newcommand{\sort}[1]{\m{#1}}
\newcommand{\Bool}{\sort{Bool}}

\newcommand{\CO}[1]{[\h#1\h]} 
\newcommand{\ECO}[2]{\exists #1.\ #2}

\newcommand{\CTerm}[4]{%
\ifthenelse{\equal{#1}{} \and \equal{#2}{} \and \equal{#3}{} \and \equal{#4}{}}%
{\mathrm{\Pi}\, X.\ s~\CO{\ECO{\vec{x}}{\varphi}}}%
{%
\ifthenelse{\equal{#1}{}}{}{\Pi #1.\ }%
#2%
\ifthenelse{\equal{#4}{}}%
{}%
{\ifthenelse{\equal{#3}{}}{~\CO{#4}}{~\CO{\ECO{#3}{#4}}}}%
}%
}

\newcommand{\CEqn}[4]{%
\ifthenelse{\equal{#1}{}}{%
\ifthenelse{\equal{#4}{}}{#2 \approx #3}{#2 \approx #3~\CO{#4}}}{%
\ifthenelse{\equal{#4}{}}{\mathrm{\Pi} #1.\, #2 \approx #3}{\mathrm{\Pi} #1.\, #2 \approx #3~\CO{#4}}%
}}

\newcommand{\CRu}[4]{%
\ifthenelse{\equal{#1}{} \and \equal{#2}{} \and \equal{#3}{} \and \equal{#4}{}}%
{\mathrm{\Pi} X.\ \ell \R r~\CO{\varphi}}%
{%
\ifthenelse{\equal{#1}{}}{}{\mathrm{\Pi} #1.\ }%
#2 \R #3%
\ifthenelse{\equal{#4}{}}%
{}%
{~\CO{#4}%
}%
}%
}

\newcommand{\R}{\rightarrow}

\newcommand{\Rs}{\stackrel{\smash{\raisebox{-.8mm}{\tiny $\sim$~}}}{\R}}

\renewcommand{\L}{\leftarrow}

\newcommand{\Lb}[1][]{\mr{\vphantom{\R}_{#1}{\L}}}

\newcommand{\Rbase}[1][]{\R_{\mathsf{base}}}

\newcommand{\Ca}[1][]{\xleftrightarrow{#1}}

\newcommand{\Cb}[1][]{\Ca[]_{#1}}

\newcommand{\Cru}[1][\xE]{\ifthenelse{\equal{#1}{}}{\Cb[\m{rule}]}{\Cb[\m{rule},#1]}}
\newcommand{\Cbase}[1][\xE]{\ifthenelse{\equal{#1}{}}{\Cb[\m{base}]}{\Cb[\m{base},#1]}}

\renewcommand{\geq}{\geqslant}

\renewcommand{\ge}{\geqslant}
\renewcommand{\le}{\leqslant}

\newcommand\subsetsim{\mathrel{\substack{
  \textstyle\subset\\[-0.2ex]\textstyle\sim}}}
\newcommand\supsetsim{\mathrel{\substack{
  \textstyle\supset\\[-0.2ex]\textstyle\sim}}}

\newcommand{\PG}{\mathrm{PG}} %_{\sim}}
\newcommand{\lvf}{\mathrm{lvf}}
\newcommand{\ext}{\mathrm{ext}}
\newcommand{\rmv}{\mathrm{rmv}}

\newcommand{\Bfnum}[1]{(#1)}

\newcommand{\crest}{\textsf{crest}}
\newcommand{\ctrl}{\textsf{Ctrl}}
\newcommand{\crisys}{\textsf{Crisys}}

\newcommand{\pvec}[1]{\vec{#1}\mkern2mu\vphantom{#1}}

\begin{document}

\title{
Recovering Commutation of Logically Constrained Rewriting and Equivalence Transformations
}

\author{Kanta Takahata}
\affiliation{%
  \institution{Niigata University}
  \city{Niigata}
  \country{Japan}
}

\author{Jonas Sch\"opf}
\orcid{0000-0001-5908-8519}
\affiliation{%
  \institution{University of Innsbruck}
  \city{Innsbruck}
  \country{Austria}
}
\email{jonas.schoepf@uibk.ac.at}

\author{Naoki Nishida}
\orcid{0000-0001-8697-4970}
\affiliation{%
  \institution{Nagoya University}
  \city{Nagoya}
  \country{Japan}
}
\email{nishida@i.nagoya-u.ac.jp}

\author{Takahito Aoto}
\orcid{0000-0003-0027-0759}
\affiliation{%
  \institution{Niigata University}
  \city{Niigata}
  \country{Japan}
}
\email{aoto@ie.niigata-u.ac.jp}

\renewcommand{\shortauthors}{Takahata et al.}

\begin{abstract}
Logically constrained term rewriting is a relatively new rewriting formalism 
that naturally supports built-in data structures, such as integers and bit vectors. 
In the analysis of logically constrained term rewrite systems (LCTRSs), 
rewriting constrained terms plays a crucial role. 
However, this combines rewrite rule applications
and equivalence transformations in a closely intertwined way.
This intertwining makes it difficult to establish useful theoretical properties 
for this kind of rewriting and causes problems in 
implementations---namely, that impractically large search spaces 
are often required.
To address this issue, we propose in this paper 
a novel notion of most general constrained rewriting, 
which operates on existentially constrained terms, 
a concept recently introduced by the authors. 
We define a class of left-linear, left-value-free LCTRSs 
that are general enough to simulate all left-linear LCTRSs
and exhibit the desired key property: 
most general constrained rewriting commutes with equivalence.
This property ensures that equivalence transformations can be deferred
until after the application of rewrite rules,
which helps mitigate the issue of large search spaces in implementations.
In addition to that, we show that the original rewriting formalism on constrained terms 
can be embedded into our new rewriting formalism on existentially constrained terms. 
Thus, our results are expected to have significant implications for achieving correct and 
efficient implementations in tools operating on LCTRSs.
\end{abstract}





\maketitle

\section{Introduction}

The basic formalism of term rewriting is a purely syntactic computational model; 
due to its simplicity, it is one of the most extensively studied computational models.
However, precisely because of this simplicity, 
it is often not suitable for applications that arise in practical areas,
such as in programming languages, formal specifications, etc.
One of the main issues of term rewriting and its real-world applications
is that the basic formalism lacks painless treatment of
built-in data structures, such as integers, bit vectors, etc.

Logically constrained term rewriting~\cite{KN13frocos} is a relatively new extension 
of term rewriting that intends to overcome such weaknesses of the basic formalism,
while keeping succinctness for theoretical analysis as a computational model.
Rewrite rules, that are used to model computations, in logically
constrained term rewrite systems (LCTRSs) are equipped with constraints over
some arbitrary theory, e.g., linear integer arithmetic.
Built-in data structures are represented via the satisfiability of
constraints within a respective theory. Implementations in LCTRS tools are
then able to check these constraints using SMT-solvers
and therefore benefit
from recent advances in the area of checking \emph{satisfiability modulo
theories}.
Recent progress on the LCTRS formalism was for example made in
confluence analysis~\cite{SM23,SMM24}, (non-)termination analysis~\cite{K16,NW18}, completion~\cite{WM18},
rewriting induction~\cite{KN14aplas,FKN17tocl}, algebraic semantics~\cite{ANS24}, and complexity analysis~\cite{WM21}.

During the analysis of LCTRSs, not only rewriting of terms but also rewriting
of constrained terms, called \emph{constrained rewriting}, is frequently used.
Here, a \emph{constrained term} consists of a term and a constraint,
which restricts the possibilities in which the term is instantiated.
For example, $\m{f}(x)~\CO{x > 2}$ is a constrained term (in LCTRS notation)
which can be intuitively considered as 
a set of terms $\SET{ \m{f}(x) \mid x > 2 }$.
Constrained rewriting is an integral part of many different analysis
techniques.
For example, 
in finding a specific joining sequence in confluence analysis
you need to
deal with two terms under a shared constraint, which results in working
on constrained terms. 
In rewriting induction, rewriting of constrained terms is used for several
of its inference steps.

Unfortunately, supporting constrained rewriting completely for LCTRSs is 
far from practical, as it involves heavy non-determinism.
The situation gets even worse as it is also equipped with equivalence transformations 
before (and/or after) each rewrite step.
\begin{example}
\label{exa:intro I}
Consider the % following 
rewrite rule
$
\rho\colon \m{f}(x) \to \m{g}(y)~ [x \ge \m{1} \land x + \m{1} \ge y]
$
and a constrained term $\m{f}(x)~\CO{x > \m{2}}$.
Because, for any $x > \m{2}$, the instantiation $y := \m{3}$ satisfies the constraint 
$x \ge \m{1} \land x + \m{1} \ge y$ of the rewrite rule,
we obtain a rewrite step
\[
\m{f}(x)~\CO{x > \m{2}} \R_\rho \m{g}(\m{3})~\CO{x > \m{2}}.
\]
It is also possible to apply the following rewrite step,
because the instantiation $y := x$ satisfies the constraint:
\[
\m{f}(x)~\CO{x > \m{2}} \R_\rho \m{g}(x)~\CO{x > \m{2}}
\]
Actually, by using equivalence transformations (denoted by $\sim$), 
different variations of rewrite steps are possible, e.g.:
\[
\begin{array}{l@{\>}c@{\>}l}
\m{f}(x)~\CO{x > \m{2}} 
&\sim&  \m{f}(x)~\CO{x > \m{2} \land \m{0} > y}\\
&\R_\rho& \m{g}(y)~\CO{x > \m{2} \land \m{0} > y}\\
&\sim&  \m{g}(y)~\CO{\m{0} > y}\\
\end{array}
\]
Note that the resulting constrained terms are often not equivalent,
e.g.\
$\m{g}(\m{3})~\CO{x > \m{2}} \not\sim \m{g}(x)~\CO{x > \m{2}}$,
as
$\SET{ \m{g}(\m{3}) } \neq \SET{ \m{g}(x) \mid x > \m{2} }$.
\end{example}
The question may arise whether restricting to rewriting
without equivalence transformations is a good idea.
However, it turns out that for some natural computations, we need 
an equivalence transformation prior to the actual rewrite step:
\begin{example}
\label{exa:intro II}
Consider the rewrite rule $\rho\colon \m{h}(x,y) \R \m{g}(z)~\CO{(x + y) + \m{1} = z}$
and a  constrained term $\m{h}(x, y)~\CO{x < y}$.
It is not possible to take any concrete value or variable of $x,y$ for $z$,
and hence the constrained term using the rule $\rho$ cannot be rewritten.
However, %if we perform an initial equivalence transformation, 
after the equivalence transformation,
the rule becomes applicable: 
\[
\begin{array}{l@{\>}c@{\>}l}
\m{h}(x,y)~\CO{x < y} 
&\sim& \m{h}(x,y)~\CO{x < y \land x + y + \m{1} = z}\\
&\R_\rho& \m{g}(z)~\CO{x < y \land x + y + \m{1} = z}
\end{array}
\]
\end{example}

Clearly, it is not feasible to support the full strength of such an equivalence relation
in an implementation.

In this paper,
we introduce
a novel notion of most general constrained rewriting, 
which operates on existentially constrained terms, 
a concept recently introduced by the authors~\cite{TSNA25LOPSTR-arxiv}.
As seen in the example above,
a key source of confusion is that the rewrite step heavily relies on 
variables in the constraint that do not appear in the term itself.
The existentially constrained terms distinguish
variables that appear solely in the constraint but not in the term itself by using
existential quantifiers. 
Variables appearing only within the constraint are naturally not allowed
to appear in any reduction of a term as they are bound to the scope of the constraint.
It turns out that this novel way of rewriting covers 
the ``most general part'' of the original rewrite relation, which in practice usually suffices 
for the analysis of LCTRSs. 

Additionally, it fulfills not only a form of uniqueness of reducts but also
the commutation property of rewrite steps and equivalence transformations. These
features are not supported by the original rewrite relation.
The latter property about commutation is very important from 
an implementation perspective, because as a result one can 
move intermediate equivalence transformations in rewrite sequences
to the end of the sequence.
This property reduces huge search spaces for the computations of rewrite sequences.
Coincidentally, LCTRS tools such as \ctrl~\cite{KN15},
\crisys~\cite{KNM25jlamp}, and crest~\cite{SMM24,SM25}
already implement similar approaches to deal with constrained rewriting.
However, this has not formally been defined so far. 
Our results
guarantee the correctness of the approaches in these
implementations and provide the foundation for their correctness.

The remainder of the paper is organized as follows.
After presenting the necessary background in \Cref{sec:preliminaries},
we introduce most general rewrite steps
and prove its well-definedness in \Cref{sec:constrained rewriting}.
In \Cref{sec:Simulating Most General Constrained Rewriting in Non-Quantfied Constrained Rewriting,sec:comparison with previous format}
we focus on the relation between our new formalism of constrained rewriting 
and the current state-of-the-art.
Then in \Cref{sec:commutativity of rewriting and equivalence},
we show two useful properties of most general rewriting:
uniqueness of reducts and commutation between 
rewrite steps and equivalence transformations on pattern-general existentially constrained terms.
Subsequently we introduce left-value-free rules~\cite{Kop17} within our setting and discuss their rewriting behavior
in \Cref{sec:left-value-free rules}.
In \Cref{sec:general commutativity left-value-free rules}
we show a general commutation theorem between rewrite steps using left-value-free rules 
and the equivalence transformation.
Before we conclude in \Cref{sec:conclusion},
we discuss some related work in \Cref{sec:Related Work}.

\section{Preliminaries}
\label{sec:preliminaries}

In this section, we briefly recall the basic notions of LCTRSs~\cite{KN13frocos,SM23,SMM24,ANS24}
and fix additional notations used throughout this paper.
Familiarity with the basic notions of term rewriting is assumed (e.g.\ see~\cite{BN98,O02}).

\paragraph{Logically Constrained Terms}

Our signature consists of
a set $\xS$ of sorts and a set $\xF$ of function symbols, where each 
$f \in \xF$ is attached with its sort declaration 
$f\colon \sig{\tau} \to \tau_0$. 
As in~\cite{ANS24}, we assume that these sets can be partitioned into
two disjoint sets, i.e., $\xS = \xSTh \uplus \xSTe$ and $\xF = \xFTh \uplus \xFTe$,
where each $f\colon \sig{\tau} \to \tau_0 \in \xFTh$ satisfies 
$\tau_i \in \xSTh$ for all $0 \leqslant i \leqslant n$. 
Elements of $\xSTh$ ($\xFTh$) and $\xSTe$ ($\xFTe$)
are called theory sorts (symbols) and term sorts (symbols).
The sets of variables and terms are denoted by $\xV$ and
$\xT(\xF,\xV)$. 
We assume a special sort $\Bool \in \xSTh$, and 
call the terms in $\xT(\xFTh,\xV)^\Bool$ \emph{logical constraints}.
We denote the set of variables appearing in terms $\seq{t}$ by $\xV(\seq{t})$.
Sometimes sequences of variables and terms are written as $\vec{x}$ and $\vec{t}$.
The set of variables occurring in $\vec{x}$ is denoted by $\SET{\vec{x}}$.
The set of sequences of elements of a set $T$ is denoted by $T^*$,
such that $\vec{x} \in \xV^*$.

The set of positions in a term $t$ is denoted by $\Pos(t)$.
The symbol and subterm occurring at a position $p \in \Pos(t)$ is denoted
by $t(p)$ and $t|_p$, respectively. For $U \subseteq \xF \cup \xV$,
we write $\Pos_U(t) = \SET{p \in \Pos(t) \mid t(p) \in U}$
for positions with symbols in $U$.
A term obtained from $t$ by replacing subterms at positions $\seq{p}$ 
by the terms $\seq{t}$, having the same sort as $t|_{p_1},\ldots,t|_{p_n}$,
is written as $t[\seq{t}]_{\seq{p}}$ or just $t[\seq{t}]$ when no confusions arises.
Sometimes we consider an expression obtained by
replacing those subterms $t|_{p_1},\ldots,t|_{p_n}$ in $t$ by holes of the same sorts,
which is called a multihole context and denoted by $t[~]_{\seq{p}}$.

A sort-preserving function $\sigma$ from $\xV$ to $\xT(\xF,\xV)$ is called
a substitution, where it is identified with its homomorphic extension 
$\sigma\colon \xT(\xF,\xV) \to \xT(\xF,\xV)$.
The domain of a substitution $\sigma$ is 
given by $\Dom(\sigma) = \SET{ x \in \xV \mid x \neq \sigma(x)}$.
A substitution $\sigma$ is written as $\sigma\colon U \to T$ if 
$\Dom(\sigma) \subseteq U$ and 
$\sigma(U) = \SET{\sigma(x) \mid x \in U} \subseteq T$.
For a set $U \subseteq \xV$, a substitution $\sigma|_U$ is given
by $\sigma|_U(x) = \sigma(x)$ if $x \in U$ and $\sigma|_U(x) = x$ otherwise.
For substitutions $\sigma_1,\sigma_2$ such that
$\Dom(\sigma_1) \cap \Dom(\sigma_2) = \varnothing$,
the substitution $\sigma_1 \cup \sigma_2$
is given by $(\sigma_1 \cup \sigma_2)(x) = \sigma_i(x)$ if $x \in \Dom(\sigma_i)$
and $(\sigma_1 \cup \sigma_2)(x) = x$ otherwise.
A substitution $\sigma\colon \SET{\seq{x}} \to \SET{\seq{t}}$
such that $\sigma(x_i)= t_i$ is denoted by $\SET{x_1 \mapsto t_1,\ldots,x_n \mapsto t_n}$;
for brevity sometimes we write just $\left\{\vec{x} \mapsto \vec{t}\; \right\}$.
A bijective substitution $\sigma\colon \xV \to \xV$ is called a renaming,
and its inverse is denoted by $\sigma^{-1}$.

A model over a sorted signature $\langle \xSTh,\xFTh \rangle$
consists of the two interpretations $\xI$ for sorts
and $\xJ$ for function symbols such that 
$\xI(\tau)$ is a non-empty set of values and
$\xJ$ assigns any function symbol $f\colon\sig{\tau} \to \tau_0 \in \xFTh$ 
to a function $\xJ(f)\colon \xI(\tau_1) \times \cdots \times \xI(\tau_n) \to \xI(\tau_0)$.
We assume a fixed model $\xM = \langle \xI, \xJ \rangle$ 
over the sorted signature $\langle \xSTh,\xFTh \rangle$
such that any element $a \in \xI(\tau)$ appears as a constant $a^\tau$ in 
$\xFTh$. These constants are called \emph{values} and the set of all values is denoted by $\Val$.
For a term $t$, we define $\Val(t) = \SET{ t(p) \mid p \in \Pos_{\Val}(t)}$
and for a substitution $\gamma$, we define $\VDom(\gamma) = \SET{x \in \xV \mid \gamma(x) \in \Val}$.
Throughout this paper we assume
the standard interpretation for the sort $\Bool \in \xSTh$,
namely $\xI(\Bool) = \mathbb{B} = \SET{\mathsf{true}, \mathsf{false}}$,
and the existence of necessary standard theory function symbols 
such as $\neg$, ${\land}$, ${\Rightarrow}$, ${=^\tau}$, etc. 

A valuation $\rho$ on the model $\xM = \langle \xI, \xJ \rangle$
is a mapping that assigns any $x^\tau \in \xV$ with $\rho(x) \in \xI(\tau)$.
The interpretation of a term $t$ in the model $\xM$ over the valuation $\rho$
is denoted by $\inter{t}_{\xM,\rho}$.
For a logical constraint $\varphi$ and %For 
a valuation $\rho$ over the model $\xM$, 
we write $\vDash_{\xM,\rho} \varphi$ if $\inter{\varphi}_{\xM,\rho} = \mathsf{true}$, 
and $\vDash_{\xM} \varphi$ if 
$\vDash_{\xM,\rho} \varphi$ for all valuations $\rho$.
For $X \subseteq \xV$, a substitution $\gamma$ is said to be $X$-valued
if $\gamma(X) \subseteq \Val$.
We write $\gamma \vDash_\xM \varphi$ (and say $\gamma$ respects $\varphi$)
if the substitution $\gamma$ is $\Var(\varphi)$-valued (or
equivalently, $\Var(\varphi) \subseteq \VDom(\gamma)$,
as well as $\gamma(\Var(\varphi)) \subseteq \Val$)
and $\vDash_{\xM} \varphi\gamma$.
If no confusion arises then we drop the subscript $\xM$ in these notations.

\paragraph{Rewriting Logically Constrained Terms}

The main focus in this paper lies on rewriting constrained terms,
hence we focus solely on introducing this concept.
A \emph{constrained term} is a tuple of a term $s$ and a logical constraint $\varphi$ 
written as $s~\CO{\varphi}$. 
Two constrained terms $s~\CO{\varphi}$, $t~\CO{\psi}$ are equivalent,
written as $s~\CO{\varphi} \sim t~\CO{\psi}$,
if for all substitutions $\sigma$ with $\sigma \vDash \varphi$ there
exists a substitution $\gamma$ with $\gamma \vDash \psi$ such that $s\sigma =
t\gamma$, and vice versa. 

A constrained rewrite rule is a triple $\ell \R r~\CO{\varphi}$ of terms $\ell,r$ with the same sort,
and a logical constraint $\varphi$. 
The set of \emph{logical variables} of the rule is
given by $\LVar(\ell \R r~\CO{\varphi}) = (\Var(r) \setminus \Var(\ell))\cup \Var(\varphi)$.
A \emph{logically constrained term rewrite system} (LCTRS, for short)
consists of a signature $\Sigma = \langle \xSTh \uplus \xSTe, \xFTh \uplus \xFTe \rangle$,
a model $\xM$ over $\SigmaTh = \langle \xSTh, \xFTh \rangle$ (which induces the
set $\Val \subseteq \xFTh$ of values) and a set $\xR$ of constrained rewrite rules over
the signature $\Sigma$. 

Consider a constrained term $s~\CO{\varphi}$ with a satisfiable logical constraint. 
Let $\ell \R r~\CO{\psi}$ be a constrained rewrite rule. 
If there exists a position $p \in \FPos(s)$, 
a substitution $\sigma$ with $\ell\sigma = s|_p$, 
$\sigma(x) \in \Val \cup \Var(\varphi)$ for all $x \in \LVar(\ell \R r~\CO{\psi})$, 
and $\vDash_\xM \varphi \Rightarrow \psi\sigma$,
then we have the rewrite step 
$s~\CO{\varphi} \R s[r\sigma]_p~\CO{\varphi}$. 
The full rewrite relation $\Rs$ on constrained terms
is defined by ${\Rs} = {\sim} \cdot {\R} \cdot {\sim}$. 
In~\cite{SMM24}, it is shown
that \emph{calculation steps} defined in~\cite{KN13frocos} can be
integrated into this rewrite relation by using the constrained rules
$\xR_\textsf{ca} = \SET{ f(x_1,\ldots,x_n) \to y ~\CO{y = f(x_1,\ldots,x_n)} \mid f
\in \xFTh \setminus \Val }$, including some additional initial $\sim$-steps.

\paragraph{Existentially Constrained Terms and Their Equivalence}

We now explain existentially constrained terms 
and characterization of their equivalence based on~\cite{TSNA25LOPSTR-arxiv}.

An \emph{existential constraint} is a pair $\langle \vec{x},
\varphi \rangle$ of a sequence of variables $\vec{x}$ and a constraint
$\varphi$, written as $\ECO{\vec{x}}{\varphi}$, 
such that $\SET{\vec{x}} \subseteq \Var(\varphi)$.
The sets of \emph{free} and \emph{bound} variables of $\ECO{\vec{x}}{\varphi}$
are given by $\FVar(\ECO{\vec{x}}{\varphi}) = \Var(\varphi) \setminus \SET{\vec{x}}$ and
$\BVar(\ECO{\vec{x}}{\varphi}) = \SET{\vec{x}}$.
We write $\vDash_{\xM,\rho} \ECO{\vec{x}}{\varphi}$ if 
there exists $\vec{v} \in \Val^*$ such that $\vDash_{\xM,\rho} \varphi\kappa$,
where $\kappa = \SET{\vec{x} \mapsto \vec{v}}$.
An existential constraint $\ECO{\vec{x}}{\varphi}$ is said to be \emph{valid}, 
written as $\vDash_{\xM} \ECO{\vec{x}}{\varphi}$, if 
$\vDash_{\xM,\rho} \ECO{\vec{x}}{\varphi}$ for any valuation $\rho$.
An existential constraint $\ECO{\vec{x}}{\varphi}$ is said to be \emph{satisfiable} if 
$\vDash_{\xM,\rho} \ECO{\vec{x}}{\varphi}$ for some valuation $\rho$.
For any substitution $\sigma$,
we write $\sigma \vDash_\xM \ECO{\vec{x}}{\varphi}$ 
(and say $\sigma$ respects $\ECO{\vec{x}}{\varphi}$) 
if $\sigma(\FVar(\ECO{\vec{x}}{\varphi})) \subseteq \Val$
and $\vDash_\xM (\ECO{\vec{x}}{\varphi})\sigma$.
An \emph{existentially constrained term} is a 
triple $\langle X, s, \ECO{\vec{x}}{\varphi} \rangle$, written as $\CTerm{X}{s}{\vec{x}}{\varphi}$, 
of a set $X$ of variables, a term $s$, and an existential constraint $\ECO{\vec{x}}{\varphi}$
such that 
$\FVar(\ECO{\vec{x}}{\varphi}) \subseteq X \subseteq \Var(s)$ and
$\BVar(\ECO{\vec{x}}{\varphi}) \cap \Var(s) = \varnothing$.
Variables in $X$ are called \emph{logical variables} (of $\CTerm{X}{s}{\vec{x}}{\varphi}$).
An existentially constrained term $\CTerm{X}{s}{\vec{x}}{\varphi}$ is said to be \emph{satisfiable} if 
$\ECO{\vec{x}}{\varphi}$ is satisfiable.
An existentially constrained term $\CTerm{X}{s}{\vec{x}}{\varphi}$ is said to be
\emph{subsumed by} an existentially constrained term
$\CTerm{Y}{t}{\vec{y}}{\psi}$, denoted by $\CTerm{X}{s}{\vec{x}}{\varphi}
\subsetsim \CTerm{Y}{t}{\vec{y}}{\psi}$, if for all $X$-valued substitutions
$\sigma$ with 
$\sigma \vDash_\xM \ECO{\vec{x}}{\varphi}$ 
there exists a
$Y$-valued substitution $\gamma$ with 
$\gamma \vDash_\xM \ECO{\vec{y}}{\psi}$
such that $s\sigma = t\gamma$.
Existentially constrained terms $\CTerm{X}{s}{\vec{x}}{\varphi}$, $\CTerm{Y}{t}{\vec{y}}{\psi}$ are said to be \emph{equivalent},
denoted by $\CTerm{X}{s}{\vec{x}}{\varphi} \sim \CTerm{Y}{t}{\vec{y}}{\psi}$, 
if $\CTerm{X}{s}{\vec{x}}{\varphi} \subsetsim \CTerm{Y}{t}{\vec{y}}{\psi}$
and $\CTerm{X}{s}{\vec{x}}{\varphi} \supsetsim \CTerm{Y}{t}{\vec{y}}{\psi}$.
Throughout this paper, we will use three characterizations of equivalence 
for existentially constrained terms, presented in~\cite{TSNA25LOPSTR-arxiv}.
The first one covers variants under renaming.

\begin{restatable}[\cite{TSNA25LOPSTR-arxiv}]{proposition}{ThmCharacterizationOfEquivalenceRsingRenaming}
\label{thm: characterization of equivalence using renaming}
Let $\delta$ be a renaming.
Let $\CTerm{X}{s}{\vec{x}}{\varphi}$, $\CTerm{Y}{t}{\vec{y}}{\psi}$ be
satisfiable existentially constrained terms such that $s\delta = t$.
Then,
$\CTerm{X}{s}{\vec{x}}{\varphi} \sim \CTerm{Y}{t}{\vec{y}}{\psi}$ iff
$\delta(X) = Y$ and $\vDash_\xM (\ECO{\vec{x}}{\varphi})\delta \Leftrightarrow (\ECO{\vec{y}}{\psi})$.
\end{restatable}

The second one considers those that are \emph{pattern-general},
where an existentially constrained term $\CTerm{X}{s}{\vec{x}}{\varphi}$ is 
\emph{pattern-general} if $s$ is $X$-linear and $\Val(s) = \varnothing$~\cite{TSNA25LOPSTR-arxiv}.
Any existentially constrained term 
can be translated to equivalent pattern-general existentially constrained term
by the following translation:
$\PG(\CTerm{X}{s}{\vec{x}}{\varphi}) = \CTerm{Y}{t}{\vec{y}}{\psi}$
where
    $t = s[w_1,\ldots,w_n]_{p_1,\ldots,p_n}$
    with $\Pos_{X\cup\Val}(s) = \SET{p_1,\ldots,p_n}$ and
    pairwise distinct fresh variables $w_1,\ldots,w_n$,
    $Y = \SET{w_1,\ldots,w_n}$,
    $\SET{\vec{y}} = \SET{\vec{x}} \cup X$, and
    $\psi = (\varphi \land \bigwedge_{i=1}^n (s|_{p_i} = w_i))$.

\begin{restatable}[\cite{TSNA25LOPSTR-arxiv}]{proposition}{TheoremIIIxxviii}
\label{thm:complete characterization of equivalence for most sim-general constrained terms}
Let $\CTerm{X}{s}{\vec{x}}{\varphi}, \CTerm{Y}{t}{\vec{y}}{\psi}$ be 
satisfiable pattern-general existentially constrained terms.
Then, we have $\CTerm{X}{s}{\vec{x}}{\varphi} \sim \CTerm{Y}{t}{\vec{y}}{\psi}$ iff there exists a renaming $\rho$ such that
    $s\rho = t$,
    $\rho(X) = Y$,    
        and
    $\vDash_\xM (\ECO{\vec{x}}{\varphi})\rho \Leftrightarrow (\ECO{\vec{y}}{\psi})$.
\end{restatable}

The last characterization does not impose any
restriction on the existentially constrained terms.
In order to present this, we need a couple of notions and
results from~\cite{TSNA25LOPSTR-arxiv}.
Let $\CTerm{X}{s}{\vec{x}}{\varphi}$ be an existentially constrained term.
Then, we define a binary relation $\sim_{\Pos_{X\cup \Val}(s)}$ over the positions in $\Pos_{X\cup \Val}(s)$ as follows:
$p \sim_{\Pos_{X\cup \Val}(s)} q$
iff
$\vDash_\xM ((\ECO{\vec{x}}{\varphi}) \Rightarrow s|_p = s|_q)$.
Then, $\sim_{\Pos_{X\cup \Val}(s)}$ is an equivalence relation over the positions in $\Pos_{X\cup \Val}(s)$.
The equivalence class of a position $p \in \Pos_{X\cup\Val}(s)$ w.r.t.\ 
$\sim_{\Pos_{X\cup \Val}(s)}$ is denoted by $[p]_{\sim_{\Pos_{X\cup \Val}(s)}}$.
If it is clear from the context then we may simply denote it by $[p]_\sim$.
We further denote the representative of $[p]_\sim$ by $\hat{p}$.
Let $\CTerm{X}{s}{\vec{x}}{\varphi}$ be a satisfiable existentially constrained term.
We define
$\Pos_{\Val!}(s) = \{ p \in \Pos_{X\cup\Val}(s)
\mid$  there exists $v \in \Val$ such that
$\vDash_\xM ((\ECO{\vec{x}}{\varphi}) \Rightarrow (s|_p = v)) \}$.
For each $p \in \Pos_{\Val!}(s)$,  there exists a unique value $v$
     such that $\xM \vDash (\ECO{\vec{x}}{\varphi}) \Rightarrow (s|_p = v)$~\cite{TSNA25LOPSTR-arxiv};
we denote such a $v$ by $\Val!(p)$.
We then define a 
\emph{representative substitution} $\mu_{X}: X \to X \cup \Val$ of 
$\CTerm{X}{s}{\vec{x}}{\varphi}$ as follows:
\[
\mu_{X}(z) = 
\left\{
\begin{array}{@{\>}ll@{}}
\Val!(p) & \mbox{if $s(p) = z$ for some $p \in \Pos_{\Val!}(s)$},\\
s(\hat{p}) & \mbox{otherwise},
\end{array}
\right.
\]
where $\hat{p}$ is the representative of the equivalence class $[p]_\sim$.

\begin{restatable}[\cite{TSNA25LOPSTR-arxiv}]{proposition}{ThmCompleteCharacterizatinOfEquivalenceOfConstrainedTerms} % from PPL
\label{thm:complete characterizatin of equivalence of constrained terms}
Let $\CTerm{X}{s}{\vec{x}}{\varphi}, \CTerm{Y}{t}{\vec{y}}{\psi}$ be satisfiable existentially constrained terms.
Then, $\CTerm{X}{s}{\vec{x}}{\varphi} \sim \CTerm{Y}{t}{\vec{y}}{\psi}$
iff
the following statements hold:
\begin{enumerate}
     \item $\Pos_{X\cup\Val}(s) = \Pos_{Y\cup\Val}(t)$ ($= \SET{p_1,\ldots,p_n}$),
     \item there exists a renaming $\rho \colon \Var(s) \setminus X \to \Var(t) \setminus Y$ such that 
     $\rho(s[\,]_{p_1,\ldots,p_n}) = t[\,]_{p_1,\ldots,p_n}$,
     \item for any $i,j \in \SET{1,\ldots,n}$, 
         $\vDash_\xM (\ECO{\vec{x}}{\varphi}) \Rightarrow (s|_{p_i} = s|_{p_j})$
         iff 
         $\vDash_\xM (\ECO{\vec{y}}{\psi}) \Rightarrow (t|_{p_i} = t|_{p_j})$,
     \item for any $i \in \SET{1,\ldots,n}$ and $v \in \Val$, 
         $\vDash_\xM (\ECO{\vec{x}}{\varphi}) \Rightarrow (s|_{p_i} = v)$
         iff 
         $\vDash_\xM (\ECO{\vec{y}}{\psi}) \Rightarrow (t|_{p_i} = v)$,
         and
     \item
     let ${\sim} = {\sim}_{\Pos_{X\cup \Val}(s)} = {\sim}_{\Pos_{X\cup \Val}(t)}$, 
     and $\mu_X,\mu_Y$ be representative substitutions of
      $\CTerm{X}{s}{\vec{x}}{\varphi}$ and $\CTerm{Y}{t}{\vec{y}}{\psi}$, respectively, 
     based on the same representative for
     each equivalence class $[p_i]_\sim$ ($1 \le i \le n$),
     $\theta =  \SET{ \langle s|_{p_i}, t|_{p_i} \rangle \mid 1 \le i \le n }$,
     $\tilde{X} = \hat{X} \cap X$, and $\tilde{Y} = \hat{Y}\cap Y$.
     Then,
     $\vDash_\xM (\ECO{\vec{x}}{\varphi})\mu_X\theta|_{\tilde{X}}
     \Leftrightarrow (\ECO{\vec{y}}{\psi})\mu_Y$,
     with the renaming $\theta|_{\tilde{X}} \colon \tilde{X} \to \tilde{Y}$.
 \end{enumerate}
 \end{restatable}

\section{Most General Constrained Rewriting}
\label{sec:constrained rewriting}

In this section,
we present the definition of
most general rewriting and prove its well-definedness.
Furthermore, we also show some properties needed
in later sections.
 
%NEW FORMALISM

We start by devising the notion of constrained rewrite rules tailored to our needs.
Similar to constrained equations~\cite{ANS24},
we attach the respective set of logical variables explicitly to constrained rewrite rules.
Previously to this, rewriting relied heavily on the notion of logical variables ($\LVar(\ldots)$) for this.

\begin{definition}[Constrained Rewrite Rule]
A \emph{constrained rewrite rule} is a quadruple of two terms $\ell$, $r$, a set of variables
$Z$ and a logical constraint $\pi$, 
written as $\CRu{Z}{\ell}{r}{\pi}$,
where 
$\ell$ and $r$ are of the same sort
and $Z$ satisfies
$(\Var(r) \setminus \Var(\ell)) \cup \Var(\pi)  \subseteq Z\subseteq \Var(\ell, r, \pi)$.
When no confusion arises, we abbreviate constrained rewrite rules
by \emph{constrained rules} or just \emph{rules};
sometimes we attach a label to it, as $\rho\colon \CRu{Z}{\ell}{r}{\pi}$,
in order to ease referencing it.
Let $\rho\colon \CRu{Z}{\ell}{r}{\pi}$ be a constrained rule.
Then, $\rho$ is \emph{left-linear} if $\ell$ is linear;
the set of \emph{extra-variables} of $\rho$
is given by $\ExVar(\rho) = \Var(r) \setminus \Var(\ell)$\footnote{%%
Do not confuse the notation $\ExVar(\rho)$ with $\xE\xV\m{ar}(\rho) = \Var(r) \setminus
(\Var(\ell) \cup \Var(\pi))$ from~\cite{SM23,SMM24}.}.
\end{definition}

It follows naturally, that a constrained rule  $\ell \R r~\CO{\varphi}$ as 
defined previously can easily be transformed into $\CRu{X}{\ell}{r}{\varphi}$
by taking $X = \LVar(\ell \R r~\CO{\varphi})$ ($:= \Var(\pi) \cup (\Var(r) \setminus \Var(\ell))$).
We also do not impose the constraint 
$\ell|_\epsilon \in \xFTe$ as it is done in~\cite{KN13frocos}, 
in order to deal with calculation steps via rule steps as in~\cite{SMM24}.

\begin{example}
For example, 
$\CRu{\{ x,y \}}{\m{f}(x)}{\m{g}(y)}{x \ge \m{1}}$
and 
$\CRu{\{ x \}}{\m{h}(x,y)}{\m{g}(y)}{x \ge \m{1}}$
are constrained rewrite rules.
However, none of these are: 
$\CRu{\{ x,y,z \}}{\m{f}(x)}{\m{g}(y)}{x \ge \m{1}}$ 
($z$ does not appear in $\Var(\m{f}(x), \m{g}(y), x \ge \m{1})$)
and $\CRu{\{ x \}}{\m{f}(x)}{\m{g}(y)}{x \ge \m{1}}$ 
($y \in \Var(r) \setminus \Var(\ell)$ does not appear in $\SET{x}$).
The rewrite rule 
$\rho\colon \m{f}(x) \to \m{g}(y)~ [x \ge \m{1} \land x + \m{1} \ge y]$
in \Cref{exa:intro I} is written
as $\CRu{\{ x,y \}}{\m{f}(x)}{\m{g}(y)}{x \ge \m{1} \land x + \m{1} \ge y}$.
\end{example}

It turns out that certain complications arise if we deal with non-left-linear rules,
\emph{hence in the following we consider only left-linear constrained rules and 
leave the extension to non-left-linear rules for future work}.
Before presenting the definition of most general constrained rewrite steps,
we introduce the notion of $\rho$-redex for a rule $\rho$.

\begin{definition}[$\rho$-Redex]
Let $\CTerm{X}{s}{\vec{x}}{\varphi}$ be a satisfiable existentially constrained term. Suppose
that $\rho\colon \CRu{Z}{\ell}{r}{\pi}$ is a left-linear constrained rule
satisfying $\Var(\rho) \cap \Var(s,\varphi) = \varnothing$. We say
\emph{$\CTerm{X}{s}{\vec{x}}{\varphi}$ has a $\rho$-redex at position $p \in
\Pos(s)$ using substitution $\gamma$} if the following is satisfied:
\textup{\Bfnum{1.}}
$\Dom(\gamma) = \Var(\ell)$,
\textup{\Bfnum{2.}}
$s|_p = \ell\gamma$,
\textup{\Bfnum{3.}}
$\gamma(x) \in \Val \cup X$ for all $x \in \Var(\ell) \cap Z$, and
\textup{\Bfnum{4.}}
$\vDash_{\xM} (\ECO{\vec{x}}{\varphi}) \Rightarrow (\ECO{\vec{z}}{\pi\gamma})$,
where $\SET{\vec{z}} = \Var(\pi) \setminus \Var(\ell)$.
\end{definition}

\begin{example}
\label{exa:redex}
Consider a constrained rule
$\rho\colon \CRu{\SET{x',y'}}{\m{f}(x')}{\m{g}(y')}{x' \ge \m{0} \land y' > x'}$.
The constrained term 
$\CTerm{\SET{x}}{\m{f}(x)}{}{x>\m{0}}$
has a $\rho$-redex at the root position using the 
substitution $\gamma = \SET{ x' \mapsto x }$.
\end{example}

\begin{definition}[Most General Rewrite Steps for Constrained Terms]
    Let $\CTerm{X}{s}{\vec{x}}{\varphi}$ be a satisfiable existentially constrained term.
    Suppose that $\rho\colon \CRu{Z}{\ell}{r}{\pi}$ is a left-linear constrained rule
    satisfying $\Var(\rho) \cap \Var(s,\varphi) = \varnothing$.
    Suppose $\CTerm{X}{s}{\vec{x}}{\varphi}$ has a $\rho$-redex 
    at position $p \in \Pos(s)$ using substitution $\gamma$.
    Then we have a rewrite step
    $\CTerm{X}{s}{\vec{x}}{\varphi} \R_\rho \CTerm{Y}{t}{\vec{y}}{\psi}$
    (or $\CTerm{X}{s}{\vec{x}}{\varphi} \R^p_{\rho,\gamma} \CTerm{Y}{t}{\vec{y}}{\psi}$
    with explicit $p$ and $\gamma$)
    where
\textup{\Bfnum{1.}}
        $t = s[r\gamma]$,
\textup{\Bfnum{2.}}
        $\psi = \varphi \land \pi\gamma$,
\textup{\Bfnum{3.}}
        $\SET{\vec{y}} = \Var(\psi) \setminus \Var(t)$, and
\textup{\Bfnum{4.}}
        $Y = \ExVar(\rho) \cup (X \cap \Var(t))$.
We call this notion \emph{most general rewrite steps}; the reason behind this naming 
will become clear in the remainder of the paper.
\end{definition}

\begin{example}[Cont'd of \Cref{exa:redex}]
\label{exa:rewriting constrained terms}
Observe that $\gamma(x') = x \in \SET{x}$
and $\vDash_\xM x > \m{0} \Rightarrow (\ECO{y'}{x \ge \m{0} \land y' > x)}$.
Therefore, we obtain a rewrite step 
\[
\begin{array}{l@{\>}c@{\>}l}
          \lefteqn{\CTerm{\SET{x}}{\m{f}(x)}{}{x>\m{0}}} \\
& \R_{\rho,\gamma} & \CTerm{\SET{y'}}{\m{g}(y')}{x}{x>\m{0} \land x \ge \m{0} \land y' > x}\\
\end{array}
\]
Note that
$(x' \ge \m{0} \land y' > x')\gamma = (x \ge \m{0} \land y' > x)$,
$\Var(x \ge \m{0} \land y' > x) \setminus \Var(\m{g}(y')) = \SET{x}$, and
$\ExVar(\rho) \cup (\SET{x} \cap \Var(\m{g}(y'))) = \SET{y'}$.
\end{example}

In a constrained rewrite step 
$\CTerm{X}{s}{\vec{x}}{\varphi} \to \CTerm{Y}{t}{\vec{y}}{\psi}$,
variables of $X$ will sometimes either be moved to $Y$ or $\SET{\vec{y}}$, or
even be eliminated such that they do not appear in $\CTerm{Y}{t}{\vec{y}}{\psi}$ anymore.
The latter behavior may happen in rewrite rules with eliminating variables,
i.e., where $\Var(\ell) \setminus \Var(r) \neq \varnothing$.
The following example illustrates such situations.

\begin{example}
Let $\rho\colon 
\CRu{\SET{v',w',x',y',z'}}{\m{f}(v',w',x',y',z')}{v'}{y' \geq \m{0}}
(= \CRu{Z}{\ell}{r}{\pi})$
be a constrained rule. Consider a rewrite step
\[
\begin{array}{@{}l@{\>}c@{\>}l@{}}
\CTerm{X}{s}{\vec{x}}{\varphi}  
&=& \CTerm{\SET{v,w,x,y,z}}{\m{g}(\m{f}(v,w,x,y,z),w)}{}{x \geq \m{0}}\\
&\R^{1}_{\rho,\gamma}& \CTerm{\SET{v,w}}{\m{g}(v,w)}{x,y}{x > \m{1} \land y \geq \m{0}}\\
&=& \CTerm{Y}{t}{\vec{y}}{\psi}
\end{array}
\]
where and $\gamma = \SET{ v' \mapsto v, \ldots, z' \mapsto z }$.
We trace now where the variables in $X$ move during this rewrite step:
\begin{itemize}
\item 
We have $v \in Y$ because
$v' \in\Var(r)$,
$\gamma(v') = v$ and
$v \in \Var(t)$.
\item 
We have $w \in Y$ because
$s|_2 = t|_2 = w$, hence
$w \in \Var(t)$.
\item 
We have $x \in \SET{\vec{y}}$. Note %that
$x \in \Var(s)$, but 
$x \notin \Var(r\gamma)$ and $x \notin \Var(t)$;
however, since $x \in\Var(\varphi)$, 
we have $x \in \SET{\vec{y}}$.
\item
We have $y \in \SET{\vec{y}}$;
note that
$y \in \Var(s)$, but 
$y \notin \Var(r\gamma)$, $y \notin \Var(t)$
and $y \notin \Var(\varphi)$;
however, $y \in \Var(\pi\gamma)$ 
and hence $y \in \Var(\psi)$.
\item
We have $z \notin Y \cup \SET{\vec{y}}$;
although $z \in \Var(s)$, we have that
$z \notin \Var(t)$,
$z \notin \Var(r)$ and 
$z \notin \Var(\psi)$.
\end{itemize}
\end{example}

In the following we show that rewriting existentially constrained terms is
well-defined, in the sense that each rewrite step results in an existentially
constrained term. Before that we give certain characterizations which are used in 
later proofs; proofs of these lemmas are presented in the appendix.

\begin{restatable}[Characterization of Free Variables in Reducts]{lemma}{LemmaIVvii}
\label{lem:free variables of constraint of reducts}
Let $\CTerm{X}{s}{\vec{x}}{\varphi}$ be a satisfiable existentially constrained term.
Suppose that $\rho\colon \CRu{Z}{\ell}{r}{\pi}$ is a left-linear constrained rule
satisfying $\Var(\rho) \cap \Var(s,\varphi) = \varnothing$.
If $\CTerm{X}{s}{\vec{x}}{\varphi} \R_{\rho} \CTerm{Y}{t}{\vec{y}}{\psi}$
then $\FVar(\ECO{\vec{y}}{\psi}) \subseteq \ExVar(\rho) \cup (X \cap \Var(t))$.
\end{restatable}

\begin{restatable}[Characterization of Bound Variables]{lemma}{LemmaIVviii}
\label{lem:bound variables of constraint of reducts}
Suppose the satisfiable existentially constrained term
$\CTerm{X}{s}{\vec{x}}{\varphi}$
such that $\CTerm{X}{s}{\vec{x}}{\varphi} \R_{\rho,\gamma} \CTerm{Y}{t}{\vec{y}}{\psi}$,
where 
$\rho\colon \CRu{Z}{\ell}{r}{\pi}$ 
is a left-linear constrained rewrite rule satisfying $\Var(\rho) \cap \Var(s,\varphi) = \varnothing$.
Then the following statements hold:
\begin{enumerate}
    \item 
$\BVar(\ECO{\vec{x}}{\varphi}) \subseteq \BVar(\ECO{\vec{y}}{\psi})$.
    \item 
$\BVar(\ECO{\vec{x}}{\varphi}) \cap \Var(l\gamma, r\gamma, \pi\gamma) = \varnothing$.
    \item 
$Y \cup \BVar(\ECO{\vec{y}}{\psi}) = \ExVar(\rho) \cup (X \cap \Var(t)) \cup \Var(\psi)$.
\end{enumerate}
\end{restatable}

We are ready to show that our rewrite steps are well-defined.

\begin{restatable}[Well-Definedness of Rewrite Steps]{theorem}{TheoremIVix}
Let $\rho$ be a left-linear constrained rewrite rule
and $\CTerm{X}{s}{\vec{x}}{\varphi}$ a satisfiable existentially constrained term
such that $\Var(\rho) \cap \Var(s,\varphi) = \varnothing$.
If $\CTerm{X}{s}{\vec{x}}{\varphi} \R_{\rho} \CTerm{Y}{t}{\vec{y}}{\psi}$
then $\CTerm{Y}{t}{\vec{y}}{\psi}$ is a %n 
satisfiable
existentially constrained term.
\end{restatable}

\begin{proof}
Suppose
$\CTerm{X}{s}{\vec{x}}{\varphi} \R^p_{\rho,\gamma} \CTerm{Y}{t}{\vec{y}}{\psi}$,
where $\rho\colon \CRu{Z}{\ell}{r}{\pi}$.
Then we have 
(1) $\Dom(\gamma) = \Var(\ell)$,
(2) $s|_p = \ell\gamma$,
(3) $\gamma(x) \in \Val \cup X$ for any $x \in \Var(\ell) \cap Z$, 
(4) $\vDash_\xM (\ECO{\vec{x}}{\varphi}) \Rightarrow (\ECO{\vec{z}}{\pi\gamma})$,
where $\SET{\vec{z}} = \Var(\pi) \setminus \Var(\ell)$,
and 
$t = s[r\gamma]$,
$\psi = \varphi \land \pi\gamma$,
$\SET{\vec{y}} = \Var(\psi) \setminus \Var(t)$, and
$Y = \ExVar(\rho) \cup (X \cap \Var(t))$.
Clearly, $\ECO{\vec{y}}{\psi}$ is an existential quantified constraint.

We show 
$\FVar(\ECO{\vec{y}}{\psi}) \subseteq Y \subseteq \Var(t)$
and $\BVar(\ECO{\vec{y}}{\psi}) \cap \Var(t) =  \varnothing$.
From \Cref{lem:free variables of constraint of reducts},
$\FVar(\ECO{\vec{y}}{\psi}) \subseteq \ExVar(\rho) \cup (X \cap \Var(t)) = Y$ follows.
Since $\SET{\vec{y}} = \Var(\psi) \setminus \Var(t)$,
$\SET{\vec{y}} \cap \Var(t) = \varnothing$ clearly holds.
It remains to show $Y = \ExVar(\rho) \cup (X \cap \Var(t)) \subseteq \Var(t)$.
However, as $X \cap \Var(t) \subseteq \Var(t)$, it suffices to show
$\ExVar(\rho) \subseteq \Var(t)$.
Let $x \in \ExVar(\rho) = \Var(r) \setminus\Var(\ell)$.
Then we obtain by $\gamma(x) = x$ that $x \notin \Var(\ell) = \Dom(\gamma)$.
Hence, from $x \in \Var(r)$, it follows that $x \in \Var(r\gamma) \subseteq \Var(t)$.

Finally, we show that $\CTerm{Y}{t}{\vec{y}}{\psi}$ is satisfiable.
By the satisfiability of $\CTerm{X}{s}{\vec{x}}{\varphi}$,
there exists a valuation $\xi$ such that $\vDash_{\xM,\xi} \ECO{\vec{x}}{\varphi}$.
Hence by (4), $\vDash_{\xM,\xi} \ECO{\vec{z}}{\pi\gamma}$.
Since $\SET{\vec{x}} \subseteq \Var(\varphi)$ and $\SET{\vec{z}} \subseteq \Var(\rho)$,
we have $\SET{\vec{x}} \cap \SET{\vec{z}} = \varnothing$.
Thus, one can extend the valuation $\xi$ to $\xi'$
such that $\vDash_{\xM,\xi'} \varphi \land \pi\gamma$.
This also gives $\vDash_{\xM,\xi'} \ECO{\vec{y}}{\psi}$.
Hence it follows that $\CTerm{Y}{t}{\vec{y}}{\psi}$ is satisfiable.
\end{proof}

\begin{remark}
In our new notion an instance of the constraint of the used rewrite rule is
attached to the constraint of the constrained term. Previously, the constraint
remained unchanged~\cite{KN13frocos}.
This implies that in each rewrite step a fresh variant of a constrained rewrite rule is needed
in order to prevent variable clashes.
Therefore, the constraint part of a constrained term 
may grow along rewriting sequences.
However, we expect that this does not cause any troubles in actual
implementations as rewrite steps can be combined with simplification of constraints.
For example, in \Cref{exa:rewriting constrained terms} one may perform simplifications
after the rewrite step:
\[
\begin{array}{@{}l@{\>}c@{\>}l@{}}
           \lefteqn{\CTerm{\SET{x}}{\m{f}(x)}{}{x>\m{0}}} \\
&\R_{\rho} & \CTerm{\SET{y'}}{\m{g}(y')}{x}{x>\m{0} \land x \ge \m{0} \land y' > x}\\
&\sim  & \CTerm{\SET{y'}}{\m{g}(y')}{}{y'>\m{1}}
\end{array}
\]
In the following, as we focus on the theoretical aspects of rewrite steps, 
we do not consider simplifications of constraints.
\end{remark}

\section{Simulating the State-of-the-Art of Non-Quantified Constrained Rewriting}
\label{sec:Simulating Most General Constrained Rewriting in Non-Quantfied Constrained Rewriting}

Our new formalism for constrained rewriting %is the
is closely related to the original definition which was introduced
in~\cite{KN13frocos}. 
In this and the following sections, we formally describe this
relation. 
To reflect the idea behind the original constrained terms of~\cite{KN13frocos}
and to prepare it for a suitable comparison to our existentially constrained terms,
we introduce the concept of \emph{non-quantified} 
constrained terms extended with the $\Pi$-notation as follows:

\begin{definition}[Non-Quantified Terms Extended with $\Pi$-notation]
A \emph{non-quantified term (extended with $\Pi$-notation)}
is a triple $\langle X,s,\varphi \rangle$,
written as $\CTerm{X}{s}{}{\varphi}$, of a set $X$ of theory variables, 
a term $s$, and a logical constraint $\varphi$ such that 
$\Var(\varphi) \subseteq X \subseteq \Var(\varphi,s)$.
\end{definition}

\begin{example}
\label{exp:non-quantified constrained terms}
For example, $\CTerm{\SET{x,y}}{\m{f}(x)}{}{x < \m{2} \land \m{0} > y}$
and $\CTerm{\SET{x,y,z}}{\m{h}(x,y)}{}{x < y \land x + y + \m{1} = z}$
are non-quantified constrained terms.
\end{example}

For non-quantified constrained terms,
satisfiability and equivalence are
defined in the usual way. A constrained term $s~\CO{\varphi}$ without
$\Pi$-notation can be easily lifted to one with $\Pi$-notation as
$\CTerm{\Var(\varphi)}{s}{}{\varphi}$.
Throughout this paper, we disambiguate the two notions of constrained terms
by explicitly naming them \emph{non-quantified constrained terms} and 
existentially constrained terms. 

% EXPLAIN ANS24
We focus now on rewrite steps of 
non-quantified constrained terms (extended with $\Pi$-notation), 
which reflects the original notion of
rewriting constrained terms~\cite{KN13frocos}.
For equational theories, Aoto et al.\ in~\cite{ANS24} introduced special $\Pi$-notations
in order to have explicit sets of logical variables which the ones that need
to be instantiated by values. The definition of constrained equation
$\CEqn{Z}{\ell}{r}{\pi}$ is very similar to $\CRu{Z}{\ell}{r}{\pi}$, albeit a
slight difference is added to the definition of constrained equations 
in~\cite{ANS24}. In particular, we request that $Z \subseteq \Var(\ell,r,\pi)$
for a constrained rule $\CRu{Z}{\ell}{r}{\pi}$ in order to avoid
redundant variables in $Z$. The variables in $Z$ of $\CRu{Z}{\ell}{r}{\pi}$
and $\CEqn{Z}{\ell}{r}{\pi}$ are the ones that need to be instantiated by
values in rewrite steps. Rewrite steps using a constrained rewrite
rule $\rho\colon \CRu{Z}{\ell}{r}{\pi}$ on a non-quantified constrained
term are performed as follows: 

\begin{definition}[Constrained Rewriting of Non-Quantified Terms]
Let $\CTerm{X}{s}{}{\varphi}$ be a satisfiable non-quantified constrained term
and suppose $\Var(\rho) \cap (X \cup \Var(s,\varphi)) = \varnothing$.
Then we obtain the rewrite step
$\CTerm{X}{s}{}{\varphi} \R^p_{\rho,\sigma} \CTerm{Y}{t}{}{\varphi}$
if there exists a position $p \in \FPos(s)$, a substitution $\sigma$ with
$\Dom(\sigma) = \Var(\ell,r,\pi)$,
$\ell\sigma = s|_p$, $\sigma(x) \in \Val \cup X$ for all $x \in Z$, and
$\vDash_\xM \varphi \Rightarrow \pi\sigma$,
where $t = s[r\sigma]_p$ and $Y = X \cap \Var(t, \varphi)$.
\end{definition}

It is %also 
easy to check that non-quantified constrained rules in~\cite{KN13frocos} 
are covered by taking 
$X = Y = \Var(\varphi)$ and $Z := \LVar(\ell \to r~\CO{\pi})$.

\begin{example}%[Cont'd from \Cref{exp:non-quantified constrained terms}]
\label{exp:rewrite step of non-quantified constrained terms}
We revisit \Cref{exa:intro I,exa:intro II}.
Consider the constrained rewrite rule
$\rho\colon 
\CRu{\SET{x',y'}}{\m{f}(x')}{\m{g}(y')}{x' \ge \m{1} \land x' + \m{1} \ge y'}$.
Then the rewrite steps in \Cref{exa:intro I} are
encoded as follows.

\begin{itemize}
\item 
Take $p = \varepsilon$ and $\sigma = \{ x' \mapsto x, y' \mapsto \m{3} \}$.
We obtain the rewrite step
$\CTerm{\SET{x}}{\m{f}(x)}{}{x > \m{2}}
\R^p_{\rho,\sigma} 
\CTerm{\SET{x}}{\m{g}(\m{3})}{}{x > \m{2}}$.
Note here that
$\sigma(x') = x \in \Val \cup X$ and $\sigma(y') = \m{3} \in \Val \cup X$.
Also,
we have $(x' \ge \m{1} \land x' + \m{1} \ge y')\sigma = (x \ge \m{1} \land x + \m{1} \ge \m{3})$
and
$\vDash_\xM (x > \m{2} \Rightarrow (x \ge \m{1} \land x + \m{1} \ge \m{3})$.

\item 
Take $p = \varepsilon$ and $\sigma = \{ x' \mapsto x, y' \mapsto x \}$.
We obtain the rewrite step
$\CTerm{\SET{x}}{\m{f}(x)}{}{x > \m{2}}
\R^p_{\rho,\sigma} 
\CTerm{\SET{x}}{\m{g}(x)}{}{x > \m{2}}$.
Note here that $\sigma(x') = \sigma(y') = x \in \Val \cup X$.
Also,
we have $(x' \ge \m{1} \land x' + \m{1} \ge y')\sigma = (x \ge \m{1} \land x + \m{1} \ge x)$
and
$\vDash_\xM (x > \m{2} \Rightarrow (x \ge \m{1} \land x + \m{1} \ge x)$.

\item
We have a rewrite step 
$\CTerm{\SET{x,y}}{\m{f}(x)}{}{x > \m{2} \land \m{0} > y}
\R^p_{\rho,\sigma} 
\CTerm{\SET{x,y}}{\m{g}(y)}{}{x > \m{2} \land \m{0} > y}$,
where $p = \varepsilon$ and $\sigma = \{ x' \mapsto x, y' \mapsto y \}$. 
Note here that 
$\sigma(x') = x \in \Val \cup X$ and $\sigma(y') = y \in \Val \cup X$.
Also,
we have $(x' \ge \m{1} \land x' + \m{1} \ge y')\sigma = (x \ge \m{1} \land x + \m{1} \ge y)$
and
$\vDash_\xM (x > \m{2} \land \m{0} > y \Rightarrow (x \ge \m{1} \land x + \m{1} \ge y)$.
\end{itemize}
Similarly, 
let $\rho\colon \CRu{\SET{x',y',z'}}{\m{h}(x',y')}{\m{g}(z')}{(x' + y') + \m{1} = z'}$
from which we obtain
$\CTerm{\SET{x,y,z}}{\m{h}(x,y)}{}{x < y \land x + y + \m{1} = z}
\R^p_{\rho,\sigma} 
\CTerm{\SET{x,y,z}}{\m{g}(z)}{}{x < y \land x + y + \m{1} = z}$
as in \Cref{exa:intro II}.
\end{example}

Below, we show the relation between the current state-of-the-art of constrained
rewrite steps and our new definition. 
We begin by defining translations between them.

\begin{definition}[Existential Extension and Existential Removing Translations]\ %hfill
\label{def:existential-extension-and-existential-removing-translations}
\begin{enumerate}
\item
An \emph{existential extension} $\ext$ of a non-quantified constrained term
is defined as $\ext(\CTerm{X}{s}{}{\varphi}) = \CTerm{Y}{s}{\vec{x}}{\varphi}$
where $\SET{\vec{x}} = \Var(\varphi) \setminus \Var(s)$ and $Y = X \setminus \SET{\vec{x}}$.
\item 
An \emph{existential removing} $\rmv$ of an existentially constrained term
is defined as $\rmv(\CTerm{X}{s}{\vec{x}}{\varphi}) = \CTerm{X \cup \SET{\vec{x}}}{s}{}{\varphi}$.
\end{enumerate}
\end{definition}

\begin{example}[Cont'd of \Cref{exp:non-quantified constrained terms}]
\label{exa:rmv and ext transformations}
The two non-quantified constrained terms in
\Cref{exp:non-quantified constrained terms}
are translated to existentially constrained terms
as follows:
$\ext(\CTerm{\SET{x,y}}{\m{f}(x)}{}{x < \m{2} \land \m{0} > y})
= \CTerm{\SET{x}}{\m{f}(x)}{\SET{y}}{x < \m{2} \land \m{0} > y}$
and 
$\ext(\CTerm{\SET{x,y,z}}{\m{h}(x,y)}{}{x < y \land x + y + \m{1} = z})
= \CTerm{\SET{x,y}}{\m{h}(x,y)}{\SET{z}}{x < y \land x + y + \m{1} = z}$.
Similarly, we have 
$\rmv(\CTerm{\SET{x}}{\m{f}(x)}{\SET{y}}{x < \m{2} \land \m{0} > y})
= \CTerm{\SET{x,y}}{\m{f}(x)}{}{x < \m{2} \land \m{0} > y}$
and 
$\rmv(\CTerm{\SET{x,y}}{\m{h}(x,y)}{\SET{z}}{x < y \land x + y + \m{1} = z})
= \CTerm{\SET{x,y,z}}{\m{h}(x,y)}{}{x < y \land x + y + \m{1} = z}$.
\end{example}

The next lemmata straightforwardly follow from \Cref{def:existential-extension-and-existential-removing-translations}.

\begin{lemma}
For any non-quantified constrained term $\CTerm{X}{s}{}{\varphi}$,
$\ext(\CTerm{X}{s}{}{\varphi})$ is an existentially constrained term;
for any existentially constrained term $\CTerm{X}{s}{\vec{x}}{\varphi}$
$\rmv(\CTerm{X}{s}{\vec{x}}{\varphi})$ is a non-quantified constrained term.
\end{lemma}

\begin{lemma}
The translation $\rmv \circ \ext$ is the identity translation on non-quantified constrained terms;
the translation $\ext \circ \rmv$ is the identity translation on existentially constrained terms.
\end{lemma}

We %first 
show that any rewrite step on
existentially constrained terms 
results in a $\Rs$-rewrite step on non-quantified
constrained terms obtained by existential removing. 
We need the following lemma.

\begin{restatable}{lemma}{LemmaIVxxx}
\label{lem:equivalence of non-quantified terms with extra constraint}
Let $\CTerm{X}{s}{}{\varphi}$ be a non-quantified constrained term.   
Suppose $\vDash_\xM \varphi \Rightarrow \ECO{\vec{z}}{\pi}$
with $\SET{\vec{z}} = \Var(\pi) \setminus (X \cup \Var(\varphi))$
and $\Var(s) \cap \SET{\vec{z}} = \varnothing$.
Then $\CTerm{X}{s}{}{\varphi} \sim \CTerm{X \cup \Var(\pi)}{s}{}{\varphi \land \pi}$.
\end{restatable}
\begin{proof}
($\subsetsim$)    
Let $\sigma$ be an $X$-valued substitution such that $\sigma \vDash_{\xM} \varphi$.
Then we have 
$X \cup \Var(\varphi) \subseteq \VDom(\sigma)$ and $\vDash_{\xM} \varphi\sigma$.
Let $\gamma$ be a valuation 
such that $\gamma(x) = \sigma(x)$ for any $x \in X \cup \Var(\varphi)$.
From the latter, we have
$\vDash_{\xM,\gamma} \varphi$ and
hence $\vDash_{\xM,\gamma} \ECO{\vec{z}}{\pi}$ by our assumption.
Thus, there exists a sequence $\vec{v} \in \Val^*$ of values so that 
$\vDash_{\xM,\gamma} \pi\kappa$, 
where $\kappa = \SET{ \vec{z} \mapsto \vec{v} }$.
As $\SET{\vec{z}} = \Var(\pi) \setminus (X \cup \Var(\varphi))$,
we have $\Var(\pi\kappa\sigma) = \varnothing$ and $\vDash_\xM \pi\kappa\sigma$.
Take a substitution $\delta = \sigma \circ \kappa$.
Then, it follows $\delta$ is $\Var(\pi)$-valued and $\vDash_\xM \varphi\delta$.
For any $x \in X$, if $x \in \Dom(\kappa)$,
then $\delta(x) = \sigma(\kappa(x)) = \kappa(x) \in \Val$.
If $x \notin \Dom(\kappa)$,
then $\delta(x) = \sigma(\kappa(x)) = \sigma(x) \in \Val$,
because $\sigma$ is $X$-valued.
Thus, $\delta$ is also $X$-valued.
By our assumption $\Var(s) \cap \SET{\vec{z}} = \varnothing$,
and thus $s\delta = s \kappa \sigma = s \sigma$.
Also, as $\SET{\vec{z}} = \Var(\pi) \setminus (X \cup \Var(\varphi))$,
we have $\SET{\vec{z}} \cap \Var(\varphi) = \varnothing$,
and therefore, 
$\varphi\delta = \varphi\kappa\sigma = \varphi\sigma$.
Hence, $\Var(\varphi\delta) = \varnothing$ and $\vDash_\xM \varphi\delta$.
%%%
All in all, $\delta$ is $X \cup \Var(\pi)$-valued, $s \sigma = s\delta$, and 
$\delta \vDash_\xM \varphi \land \pi$.
The ($\supsetsim$)-part is trivial.
\end{proof}

\begin{restatable}[Simulation of Rewrite Steps by Existential Removing]{theorem}{TheoremIVxxxi}
Let $\rho$ be a left-linear constrained rewrite rule.
If we have $\CTerm{X}{s}{\vec{x}}{\varphi} \R_\rho \CTerm{Y}{t}{\vec{y}}{\psi}$
then $\rmv(\CTerm{X}{s}{\vec{x}}{\varphi}) \Rs \rmv(\CTerm{Y}{t}{\vec{y}}{\psi})$.
\end{restatable}

\begin{proof}
Let $\rho\colon \CRu{Z}{\ell}{r}{\pi}$ be a left-linear constrained rewrite rule and 
suppose that $\CTerm{X}{s}{\vec{x}}{\varphi} \R^p_{\rho,\gamma} \CTerm{Y}{t}{\vec{y}}{\psi}$ 
such that $\Var(s,\varphi) \cap \Var(\rho) = \varnothing$.
Then we have~%
(1) $\Dom(\gamma) = \Var(\ell)$,~%
(2) $s|_p = \ell\gamma$,~%
(3) $\gamma(x) \in \Val \cup X$ for any $x \in \Var(\ell) \cap Z$, and~%
(4) $\vDash_\xM (\ECO{\vec{x}}{\varphi}) \Rightarrow (\ECO{\vec{z}}{\pi\gamma})$,
where $\SET{\vec{z}} = \Var(\pi) \setminus \Var(\ell)$,
$t = s[r\gamma]$,
$\psi = \varphi \land \pi\gamma$,
$\SET{\vec{y}} = \Var(\psi) \setminus \Var(t)$, and
$Y = \ExVar(\rho) \cup (X \cap \Var(t))$.

Note that we have 
$\rmv(\CTerm{X}{s}{\vec{x}}{\varphi}) = \CTerm{X \cup \SET{\vec{x}}}{s}{}{\varphi}$.
We begin by showing the following equivalence:
(5)
$\CTerm{X \cup \SET{\vec{x}}}{s}{}{\varphi}\\
\sim
\CTerm{X \cup \SET{\vec{x}} \cup \Var(\pi\gamma) \cup \ExVar(\rho)}
      {s}
      {}
      {\varphi \land \pi\gamma \land \bigwedge_{z \in \ExVar(\rho)} (z = z)}
$.
For this, we use \Cref{lem:equivalence of non-quantified terms with extra constraint}.
By~(4),  %$\xM \JS{\vDash} (\ECO{\vec{x}}{\varphi}) \Rightarrow (\ECO{\vec{z}}{\pi\gamma})$,
we have that $\vDash_\xM \varphi \Rightarrow \ECO{\vec{z}}{\pi\gamma}$ holds.
Since we have $\vDash_\xM \bigwedge_{z \in \ExVar(\rho)} (z = z)$,
also $\vDash_\xM \varphi \Rightarrow \ECO{\vec{z}}{\pi\gamma \land \bigwedge_{z \in \ExVar(\rho)} (z = z)}$
holds.
Now, we have
\[
\begin{array}{@{}l@{\>}c@{\>}l@{}}
\lefteqn{\textstyle \Var(\pi\gamma \land \bigwedge_{z \in \ExVar(\rho)} (z = z))
\setminus (X \cup \SET{\vec{x}} \cup \Var(\varphi))}\\
&=& ((\Var(\pi) \setminus \Var(\ell))
\cup (\bigcup_{y \in \Var(\ell) \cap \Var(\pi)} \Var(\gamma(y)))\\
&& {}\cup \ExVar(\rho))\, \setminus \, (X \cup \SET{\vec{x}} \cup \Var(\varphi))\\
&=&
((\Var(\pi) \setminus \Var(\ell))
\cup \ExVar(\rho))
\setminus (X \cup \SET{\vec{x}} \cup \Var(\varphi))\\
&=& \Var(r,\pi) \setminus \Var(\ell).
\end{array}
\]
%%%
Take 
$\SET{\pvec{z}'} = \Var(r,\pi) \setminus  \Var(\ell) =
\Var(\pi\gamma \land \bigwedge_{z \in \ExVar(\rho)} (z = z))
\setminus (X \cup \SET{\vec{x}} \cup \Var(\varphi))$.
By $\Var(\rho) \cap \Var(s) = \varnothing$,
we have $\SET{\pvec{z}'} \cap \Var(s) = \varnothing$.
Furthermore, since $\SET{\pvec{z}'} \supseteq \ExVar(\rho) = \SET{\vec{z}}$,
we have
$\vDash_\xM \varphi \Rightarrow \ECO{\pvec{z}'}{\pi\gamma \land \bigwedge_{z \in \ExVar(\rho)} (z = z)}$.
Thus, by \Cref{lem:equivalence of non-quantified terms with extra constraint},
we conclude the equivalence~(5).

We proceed to show the rewrite step
$\CTerm{X \cup \SET{\vec{x}} \cup \Var(\pi\gamma) \cup \ExVar(\rho)}
      {s}
      {}
      {\varphi \land \pi\gamma \land \bigwedge_{z \in \ExVar(\rho)} (z = z)}%\\\nonumber
\R^p_{\rho'}
\CTerm{Y \cup \SET{\vec{y}} \cup \ExVar(\rho)}
      {t}
      {}
      {\varphi \land \pi\gamma \land \bigwedge_{z \in \ExVar(\rho)} (z = z)}
$
over %the 
non-quantified constrained terms, 
by a variant $\rho'\colon \CRu{Z'}{\ell'}{r'}{\pi'}$ of the constrained rewrite rule $\rho$
such that $\Var(\rho) \cap \Var(\rho') = \varnothing$.
Accordingly, we have a variable renaming $\xi$ such 
that $Z = \xi(Z')$, $\ell = \ell'\xi$, $r = r'\xi$, and $\pi = \pi'\xi$.
We also suppose that
$\Var(\rho') \cap
\Var(s,\varphi \land \pi\gamma \land \bigwedge_{z \in \ExVar(\rho)} (z = z)) =
\varnothing$.

Take $\delta = \gamma \circ \xi$.
Then, by~(1) and our assumption on $\xi$, 
we have $\Dom(\delta) = \Var(\ell',r',\pi')$. 
Also, $\ell'\delta = \ell'\xi\gamma = \ell\gamma = s|_p$ holds
by~(2) and the definition of $\delta$.

We now show $\delta(x) \in \Val \cup X \cup \SET{\vec{x}} \cup \Var(\pi\gamma) \cup \ExVar(\rho)$
for all $x \in Z'$.
To this end assume that $x \in Z'$.
If $x \in \Var(\ell')$,
then $\delta(x) = x\xi\gamma = \gamma(y)$ for some $y \in \Var(\ell) \cap Z$.
Hence, $\delta(x) \in \Val \cup X$ by~(3).
Thus, suppose $x \in Z' \setminus \Var(\ell')$.
Then, because of $Z' \subseteq \Var(\ell',r',\pi')$, 
we know $x \in \Var(r',\pi') \setminus \Var(\ell')$.
Thus, $\delta(x) = x\xi\gamma = \gamma(y)$ for some $y \in \Var(r,\pi) \setminus \Var(\ell)$.
Hence, if $y \in \Var(\pi) \setminus \Var(\ell)$ then $\delta(x) = \gamma(y) \in \Var(\pi\gamma)$,
and if $y \in \Var(r) \setminus \Var(\ell)$ then $\delta(x) = \gamma(y) = y \in \ExVar(\rho)$.

Furthermore, we have $\vDash_\xM 
(\varphi \land \pi\gamma \land \bigwedge_{z \in \ExVar(\rho)} (z = z))
\Rightarrow \pi'\delta$, because $\pi'\delta = \pi'\xi\gamma = \pi\gamma$.
We also have $t = s[r\gamma]_{p} = s[r'\xi\gamma]_p = s[r'\delta]_p$.
Thus, it remains to show that
$Y \cup \SET{\vec{y}} \cup \ExVar(\rho)
= 
(X \cup \SET{\vec{x}} \cup \Var(\pi\gamma) \cup \ExVar(\rho)) \cap 
\Var(t, \varphi \land \pi\gamma \land \bigwedge_{z \in \ExVar(\rho)} (z = z))$.
We denote the right-hand side of the equation by $rhs$, and obtain:
$(rhs) 
= ((X \cup \SET{\vec{x}} \cup \Var(\pi\gamma) \cup \ExVar(\rho)) \cap \Var(t))
\cup  ((X \cup \SET{\vec{x}} \cup \Var(\pi\gamma) \cup \ExVar(\rho)) \cap 
\Var(\varphi \land \pi\gamma \land \bigwedge_{z \in \ExVar(\rho)} (z = z)))
=  (X \cap \Var(t)) \cup (\Var(\pi\gamma) \cap \Var(t))
\cup \Var(\varphi \land \pi\gamma) \cup \ExVar(\rho)
= 
(X \cap \Var(t)) \cup 
(\Var(\varphi \land \pi\gamma \land \bigwedge_{z \in \ExVar(\rho)} (z = z)))
\cup \ExVar(\rho)
= 
Y \cup \SET{\vec{y}} \cup \ExVar(\rho) 
$
(by \Cref{lem:bound variables of constraint of reducts}).
Putting all of this together gives the claimed rewrite step.
To complete the proof, we need to show the equivalence 
$\CTerm{Y \cup \SET{\vec{y}} \cup \ExVar(\rho)}
      {t}
      {}
      {\varphi \land \pi\gamma \land \bigwedge_{z \in \ExVar(\rho)} (z = z)}
\sim  \CTerm{Y \cup \SET{\vec{y}}}{t}{\vec{y}}{\varphi \land \pi\gamma})
= \rmv(\CTerm{Y}{t}{\vec{y}}{\psi})
$
which should be trivial to see.
\end{proof}

\begin{example}
Let us consider the most general rewrite step 
$\CTerm{\SET{x}}{\m{f}(x)}{}{x>\m{0}}
\R_{\rho} \CTerm{\SET{y'}}{\m{g}(y')}{x}{x>\m{0} \land x \ge \m{0} \land y' > x}$
of \Cref{exa:rewriting constrained terms}
where
$\rho\colon \CRu{\SET{x',y'}}{\m{f}(x')}{\m{g}(y')}{x' \ge \m{0} \land y' > x'}$.
We obtain the rewrite step over the non-quantified constrained terms
$
\rmv(\CTerm{\SET{x}}{\m{f}(x)}{}{x>\m{0}}) 
\Rs
\rmv(\CTerm{\SET{y'}}{\m{g}(y')}{x}{x>\m{0} \land x \ge \m{0} \land y'> x})
$
as follows:
\[
\begin{array}{@{}l@{\>}c@{\>}l@{}}
\lefteqn{\rmv(\CTerm{\SET{x}}{\m{f}(x)}{}{x>\m{0}})} \\
&=& \CTerm{\SET{x}}{\m{f}(x)}{}{x>\m{0}}\\
&\sim& \CTerm{\SET{x,y'}}{\m{f}(x)}{}{x>\m{0} \land x \ge \m{0} \land y' > x}\\
&\to_{\rho}& \CTerm{\SET{x,y'}}{\m{g}(y')}{}{x>\m{0} \land x \ge \m{0} \land y' > x}\\
&=&  \rmv(\CTerm{\SET{y'}}{\m{g}(y')}{x}{x>\m{0} \land x \ge \m{0} \land y'> x})
\end{array}
\]
\end{example}

\section{Embedding Non-Quantified Constrained Rewriting into the Most General Form}
\label{sec:comparison with previous format}

In this section, we concern ourselves with
the problem of characterizing the opposite direction.
That is,  whether any rewrite step on non-quantified constrained terms 
results in a rewrite step on existentially constrained terms
that are obtained by the existential extension.
Naively this is not the case as depicted by the following example.

\begin{example}
\label{exa:embedding of non-quantified rewrite step to most general rewrite steps}
In \Cref{exp:rewrite step of non-quantified constrained terms},
we encoded the rewrite steps of \Cref{exa:intro I}
by rewriting of non-quantified constrained terms with
the constrained rewrite rule
$\rho\colon \CRu{\SET{x',y'}}{\m{f}(x')}{\m{g}(y')}{x' \ge \m{1} \land x' + \m{1} \ge y'}$.
We had 
$\CTerm{\SET{x}}{\m{f}(x)}{}{x > \m{2}} \R_{\rho}  \CTerm{\SET{x}}{\m{g}(\m{3})}{}{x > \m{2}}$
and
$\CTerm{\SET{x}}{\m{f}(x)}{}{x > \m{2}} \R_{\rho} \CTerm{\SET{x}}{\m{g}(x)}{}{x > \m{2}}$.
Applying the following most general rewrite steps to
$\ext(\CTerm{\SET{x}}{\m{f}(x)}{}{x > \m{2}}) = \CTerm{\SET{x}}{\m{f}(x)}{}{x > \m{2}})$
yields
$\CTerm{\SET{x}}{\m{f}(x)}{}{x > \m{2}}
\to_\rho
\CTerm{\SET{x,y'}}{\m{g}(y')}{}{x > \m{2} \land x \ge 1 \land x + \m{1} \ge y'}$.
Unfortunately, neither
$\ext(\CTerm{\SET{x}}{\m{g}(\m{3})}{}{x > \m{2}}) = \CTerm{\SET{x}}{\m{g}(\m{3})}{}{x > \m{2}})$
nor 
$\ext(\CTerm{\SET{x}}{\m{f}(x)}{}{x > \m{2}}) = \CTerm{\SET{x}}{\m{f}(x)}{}{x > \m{2}})$
is obtained.
However, note that
$\CTerm{\SET{x,y'}}{\m{g}(y')}{}{x > \m{2} \land x \ge \m{1} \land x + \m{1} \ge y'}
\supsetsim \CTerm{\SET{x}}{\m{g}(\m{3})}{}{x > \m{2}}$,
as well as
$\CTerm{\SET{x,y'}}{\m{g}(y')}{}{x > \m{2} \land x \ge \m{1} \land x + \m{1} \ge y'}
\supsetsim \CTerm{\SET{x}}{\m{g}(x)}{}{x > \m{2}}$ holds.
\end{example}

In fact, as the example above demonstrates,
we can give the following characterization of
non-quantified constrained rewriting
in our formalism of most general rewriting.
Namely most general rewrite steps
subsume any reduct of non-quantified rewrite steps---
this relation serves as
the main motivation for us to call our new definition \emph{most general rewrite steps}.

\begin{restatable}[Simulation of Rewrite Steps by Existential Extension]{theorem}{TheoremIVxxxii}
Let $\rho$ be a left-linear constrained rewrite rule.
If we have $\CTerm{X}{s}{}{\varphi} \R_{\rho} \CTerm{Y}{t}{}{\varphi}$,
then $\ext(\CTerm{X}{s}{}{\varphi}) \R \cdot \supsetsim \ext(\CTerm{Y}{t}{}{\varphi})$.
\end{restatable}

\begin{proof}
Assume $\CTerm{X}{s}{}{\varphi} \R^p_{\rho,\gamma} \CTerm{Y}{t}{}{\varphi}$
where $\rho\colon\CRu{Z}{\ell}{r}{\pi}$ is a left-linear constrained rewrite rule
with $\Var(\rho) \cap (X \cup \Var(s,\varphi)) = \varnothing$.
Note that this is a rewrite step for non-quantified constrained terms.
Thus,  we have~%
(1) $\Dom(\gamma) = \Var(\ell,r,\pi)$,~%
(2) $s|_p = \ell\gamma$,~%
(3) $\gamma(x) \in \Val \cup X$ for all $x \in Z$, and~%
(4) $\vDash_\xM \varphi \Rightarrow \pi\gamma$.
Furthermore, $t = s[r\gamma]_p$ and $Y = X \cap \Var(t, \varphi)$.
W.l.o.g.\ we assume $\Var(\rho) \cap (Y \cup \Var(t,\varphi)) = \varnothing$.
Let $\ext(\CTerm{X}{s}{}{\varphi}) = \CTerm{X\setminus \SET{\vec{x}}}{s}{\vec{x}}{\varphi}$
with $\SET{\vec{x}} = \Var(\varphi) \setminus \Var(s)$
and %let 
$\ext(\CTerm{Y}{t}{}{\psi}) = \CTerm{Y \setminus \SET{\vec{y}}}{t}{\vec{y}}{\psi}$,
with $\vec{y} = \Var(\psi) \setminus \Var(t)$.    

We first show that
$\CTerm{X\setminus \SET{\vec{x}}}{s}{\vec{x}}{\varphi}$ has a $\rho$-redex
at position $p$ using a substitution $\sigma = \gamma|_{\Var(\ell)}$.
For this, we need to show~%
(5) $\Dom(\sigma) = \Var(\ell)$,~%
(6) $s|_p = \ell\sigma$,~%
(7) $\sigma(x) \in \Val \cup (X \setminus \SET{\vec{x}})$ for any $x \in \Var(\ell) \cap Z$, and~%
(8) $\vDash_\xM (\ECO{\vec{x}}{\varphi}) \Rightarrow (\ECO{\vec{z}}{\pi\sigma})$,
where $\SET{\vec{z}} = \Var(\pi) \setminus \Var(\ell)$.
(5) is obvious by definition of $\sigma$.
Similarly,~(6) follows as $s|_p = \ell\gamma = \ell\sigma$ by~(2).
We have $\sigma(x) \in \Val \cup X$ for any $x \in \Var(\ell) \cap Z$,
because $\sigma(x) = \gamma(x)$ for all $x \in \Var(\ell)$ and $\gamma(x) \in \Val \cup X$ by~(3);
moreover, if $\sigma(x) \in X$, 
then, as $\sigma(x)  \in \Var(\ell\sigma) \subseteq \Var(s)$ by $x \in \Var(\ell)$,
we have $\sigma(x) \in X \cap \Var(s)$, 
which implies that $\sigma(x) \notin \SET{\vec{x}} = \Var(\varphi) \setminus \Var(s)$.
Hence~(7) $\sigma(x)  \in \Val \cup (X \setminus \SET{\vec{x}})$ for any $x \in \Var(\ell) \cap Z$.

It remains to show~%
(8) $\vDash_\xM (\ECO{\vec{x}}{\varphi}) \Rightarrow (\ECO{\vec{z}}{\pi\sigma})$.
First of all, note that by definition of $\sigma$ and $\Var(\rho) \cap \Var(s) = \varnothing$, 
we have $\gamma = \gamma' \circ \sigma$ where
$\gamma' = \SET{ z \mapsto \gamma(z) \mid z \in \Var(r,\pi) \setminus \Var(\ell) }$.
Suppose $\vDash_{\xM,\xi} \ECO{\vec{x}}{\varphi}$ for a valuation $\xi$.
Then, for some $\kappa = \SET{ \vec{x} \mapsto \vec{v} }$ with $\vec{v} \in \Val^*$
we have $\vDash_{\xM,\xi \circ \kappa} \varphi$.
Hence, from~(4) 
it follows $\vDash_{\xM,\xi \circ \kappa} \pi\gamma$.
Then, by $\gamma = \gamma' \circ \sigma$,
we have $\vDash_{\xM,\xi \circ \kappa} \ECO{\vec{z}}{\pi\sigma}$
where $\SET{\vec{z}} = \Var(\pi) \setminus \Var(\ell)$.
Also, $\FVar(\ECO{\vec{z}}{\pi\sigma}) \subseteq \Var(\ell\gamma) \subseteq \Var(s)$.
Since $\Dom(\kappa) = \SET{\vec{x}} = \Var(\varphi) \setminus \Var(s)$,
we know $\Dom(\kappa) \cap \FVar(\ECO{\vec{z}}{\pi\sigma}) = \varnothing$,
and thus $(\ECO{\vec{z}}{\pi\sigma})\kappa = \ECO{\vec{z}}{\pi\sigma}$.
Hence, $\vDash_{\xM,\xi} \ECO{\vec{z}}{\pi\sigma}$ follows
and therefore we have shown that~%
(8) $\vDash_\xM (\ECO{\vec{x}}{\varphi}) \Rightarrow (\ECO{\vec{z}}{\pi\sigma})$
holds.

We conclude that $\CTerm{X\setminus \SET{\vec{x}}}{s}{\vec{x}}{\varphi}$ has a $\rho$-redex
at position $p$ using the substitution $\sigma$.
Consequently, we obtain the rewrite step 
$\CTerm{X\setminus \SET{\vec{x}}}{s}{\vec{x}}{\varphi}
\R^p_{\rho,\sigma} \CTerm{Y'}{t'}{\pvec{y}'}{\psi'}$
where
$t' = s[r\sigma]_p$,
$\psi' = \varphi \land \pi\sigma$,
$\SET{\pvec{y}'} = \Var(\psi') \setminus \Var(t')$, and
$Y' = \ExVar(\rho) \cup ((X \setminus \SET{\vec{x}}) \cap \Var(t'))$.

It remains to show
$\CTerm{Y'}{t'}{\pvec{y}'}{\psi'}
\supsetsim \CTerm{Y\setminus \SET{\vec{y}}}{t}{\vec{y}}{\psi}$.

Prior to that, let us show 
$\SET{\vec{y}} \subseteq  \SET{\pvec{y}'}$.
Let $w \in \SET{\vec{y}} = \Var(\varphi) \setminus \Var(t)$.
Since $\SET{\pvec{y}'} = \Var(\varphi \land \pi \gamma ) \setminus \Var(t')$,
it suffices to show that for any $w \in \Var(\varphi)$,
it holds that $w \in \Var(t')$ implies $w \in \Var(t)$.
Since $t' = s[r\sigma]$ and $t = s[r\gamma]$,
the case $w \in \Var(s[~])$ is trivial.
Thus, suppose $w \in \Var(r\sigma)$.
Then either $w \in \ExVar(\rho)$ or
$w \in \Var(x\sigma)$ for some $x \in \Var(\ell) \cap \Var(r)$.
The former case contradicts with $w \in \Var(\varphi)$ because $\Var(\varphi) \cap \Var(\rho) = \varnothing$.
In the latter case, $w \in \Var(r\gamma)$ follows,
as $\sigma(x) = \gamma(x)$ for $x \in \Var(\ell)$.

We proceed to show that 
$\CTerm{Y'}{t'}{\pvec{y}'}{\psi'}
\supsetsim \CTerm{Y\setminus \SET{\vec{y}}}{t}{\vec{y}}{\psi}$.
Let $\theta$ be a 
$(Y \setminus \SET{{\vec{y}}})$-valued substitution
such that $\theta \vDash_{\xM} \ECO{\vec{y}}{\psi}$,
i.e.\ 
$((\Var(\psi) \cup Y) \setminus \SET{{\vec{y}}}) \subseteq \VDom(\theta)$
and $\vDash_{\xM} (\ECO{\vec{y}}{\psi})\theta$.
We will prove there exists a substitution $\theta'$
such that it is $Y'$-valued and $\theta' \vDash \ECO{\pvec{y}'}{\psi'}$, i.e.\ 
$(Y' \cup (\Var(\psi') \setminus \SET{\pvec{y}'})) \subseteq \VDom(\theta')$
and $\vDash_\xM (\ECO{\pvec{y}'}{\psi'})\theta'$,
which satisfies $t\theta  = t'\theta'$.
For this purpose, we assume w.l.o.g.\ 
\begin{itemize}
    \item $\Dom(\theta) \cap \Var(\rho) = \varnothing$,
    \item $\Dom(\theta) \cap \SET{\vec{y}} = \varnothing$, and
    \item $\Var(\theta(z)) \cap \SET{\vec{y}} = \varnothing$ 
    for any $z \in \FVar(\ECO{\vec{y}}{\psi})$.
\end{itemize}
Note here also $\SET{\vec{y}} \cap \Var(t) = \varnothing$, as 
$\SET{\vec{y}} = \Var(\varphi) \setminus \Var(t)$.
Let us define $\theta' = \theta \cup \SET{ z \mapsto z\gamma\theta 
\mid z \in \Var(r,\pi) \setminus \Var(\ell) }$.
As a first step, we show that $\theta(\gamma(x)) = \theta'(\sigma(x))$ for
any $x \in \Var(\rho)$.
Suppose that $x \in \Var(\ell)$.
Then $\gamma(x) = \sigma(x)$ and since $x\gamma$ is a subterm of $\ell\gamma$,
we have that $\Var(\gamma(x)) \cap \Var(\rho) = \varnothing$.
In particular, $\Var(\gamma(x)) \cap (\Var(r,\pi) \setminus \Var(\ell)) = \varnothing$,
thus, by definition of $\theta'$, we have $\theta'(\gamma(x)) = \theta(\gamma(x))$.
Hence $\theta(\gamma(x)) = \theta'(\gamma(x)) = \theta'(\sigma(x))$.
Otherwise, suppose that $x \in (\Var(r,\pi) \setminus \Var(\ell))$.
Then, by definition of $\theta'$, we have $\theta'(x) = \theta(\gamma(x))$.
Since $\Dom(\sigma) = \Var(\ell)$, we have $\sigma(x) = x$ and
thus $\theta(\gamma(x)) = \theta'(x) = \theta'(\sigma(x))$.

We conclude that $\theta(\gamma(x)) = \theta'(\sigma(x))$ for any $x \in \Var(\rho)$.
From this, it follows that $\pi\gamma\theta = \pi\sigma\theta'$
and $r\gamma\theta = r\sigma\theta'$.
We further have $\theta'(x) = \theta(x)$ for any $x \in \Var(s)$
by definition of $\theta'$ as $\Var(s) \cap \Var(\rho) = \varnothing$.
Therefore, we have
$t\theta = s[r\gamma]\theta = s\theta[r\gamma\theta]
= s\theta'[r\sigma\theta'] = s[r\sigma]\theta' =  t'\theta'$.

Let us claim that $\theta'(Y') \subseteq \Val$.
Assume $y \in Y'$.
We are now going to show that $\theta'(y) \in \Val$ by distinguishing two cases.
\begin{itemize}
\item 
Assume $y \in \ExVar(\rho)$.
Then, $y \in \Var(r) \setminus \Var(\ell)$ and hence
$\theta'(y) = \theta(\gamma(y))$.
Since $\Var(r,\pi) \setminus \Var(\ell) \subseteq Z$ by the definition of $\rho$,
we have that $\gamma(y) \in \Val \cup X$ by our condition~(3).
In case $\gamma(y) \in \Val$,
clearly, $\theta'(y) = \theta(\gamma(y)) = \gamma(y) \in \Val$.
Otherwise, we have $\gamma(y) \in X$.
As $y \in \ExVar(\rho)$ and $\Dom(\gamma) = \Var(\ell,r,\pi)$,
$\gamma(y)$ appears in $t = s[r\gamma]$.
Thus, we have $\gamma(y) \in \Var(t)$.
Moreover, by $\SET{\vec{y}} = \Var(\varphi) \setminus \Var(t)$, 
we know that $\gamma(y) \notin \SET{\vec{y}}$.
Thus, $\gamma(y) \in 
(X \cap \Var(t,\varphi)) \setminus \SET{\vec{y}}
= Y \setminus \SET{\vec{y}}$.
As $\theta$ is $(Y \setminus \SET{\vec{y}})$-valued,
we conclude $\theta'(y) = \theta(\gamma(y)) \in \Val$.

\item 
Assume $y \notin \ExVar(\rho)$.
As $Y' = \ExVar(\rho) \cup ((X \setminus \SET{\vec{x}}) \cap \Var(t'))$,
we have $y \in ((X \setminus \SET{\vec{x}}) \cap \Var(t'))$.
If $y \notin \Var(t) = \Var(s[r\gamma]_p)$,
then, as $y \in \Var(t') = \Var(s[r\sigma]_p)$ and $\sigma = \gamma|_{\Var(\ell)}$,
we have $y \in \ExVar(\rho)$.
This contradicts our assumption and
therefore $y \in \Var(t)$.
Then, since $y \in X$, we have $y \in (X \cap \Var(t,\varphi)) = Y$.
Also, by $\SET{\vec{y}} = \Var(\varphi) \setminus \Var(t)$,
we know that $y \notin \SET{\vec{y}}$.
Thus, it follows that $y \in Y \setminus \SET{\vec{y}}$,
and moreover, that $\theta(y) \in \Val$ as $\theta$ is $(Y \setminus \SET{\vec{y}})$-valued.
As $y \notin \ExVar(\rho)$, we have $\theta'(y) = \theta(y)$
and we obtain $\theta'(y) \in \Val$.
\end{itemize}

It remains to show that 
$\theta' \vDash_{\xM} \ECO{\pvec{y}'}{\psi'}$.
First, 
$\FVar(\ECO{\pvec{y}'}{\psi'}) \subseteq \VDom(\theta')$
follows from $\FVar(\ECO{\pvec{y}'}{\psi'}) \subseteq Y'$
and the fact that $\theta'$ is $Y'$-valued as shown above.
We are now going to show $\vDash_{\xM} (\ECO{\pvec{y}'}{\psi'})\theta'$.
From $\vDash_{\xM} (\ECO{\vec{y}}{\varphi})\theta$,
there exists a valuation $\eta = \SET{ \vec{y} \mapsto \vec{v} }$ with $\vec{v} \in \Val^*$
such that $\vDash_{\xM} \varphi\eta\theta$.
Using the second and third assumption on $\theta$ which we have 
assumed above w.l.o.g.,
it follows that $\vDash_{\xM} \varphi\theta\eta$.
Thus, from~(4), we have that $\vDash_\xM \pi\gamma\theta\eta$,
and hence $\vDash_\xM \pi\sigma\theta'\eta$.
Again, using the second and third assumptions on $\theta$,
it follows that $\vDash_\xM \pi\sigma\eta\theta'$.
Also, from $\Var(\rho) \cap \Var(\varphi) = \varnothing$,
we have that $\varphi\theta = \varphi\theta'$,
and using our assumptions on $\theta$ 
it also follows that $\varphi\eta\theta = \varphi\eta\theta'$.
Therefore, $\vDash_\xM \varphi\eta\theta'$ holds and we conclude that
$\vDash_\xM (\varphi \land \pi\sigma)\eta\theta'$.
We obtain that $\vDash_\xM \ECO{\vec{y}}(\varphi \land \pi\sigma)\theta'$,
and $\SET{\vec{y}} \subseteq  \SET{\pvec{y}'}$ implies that
$\vDash_\xM \ECO{\pvec{y}'}(\varphi \land \pi\sigma)\theta'$.
We ultimately have $\theta' \vDash_\xM \ECO{\pvec{y}'}{\psi'}$ which concludes 
the proof.
\end{proof}

\section{Uniqueness of Reducts and Partial Commutation of Rewriting and Equivalence}
\label{sec:commutativity of rewriting and equivalence}

The aim of this section is to show 
two useful properties of most general rewriting.
First, we show that 
applying the same constrained rewrite rule at the same position yields a unique 
term.
Second, we show that our new notion of rewriting
commutes with the equivalence transformation
for pattern-general constrained terms.

\begin{restatable}[Uniqueness of Reducts]{theorem}{TheoremIVx}
Let $\CTerm{X}{s}{\vec{x}}{\varphi}$ be 
an existentially constrained term
and $p \in \Pos(s)$.
Suppose that $\rho,\rho'$ are renamed variants of the same left-linear constrained rule
satisfying $\Var(\rho) \cap \Var(s,\varphi) = \varnothing$
and $\Var(\rho') \cap \Var(s,\varphi) = \varnothing$.
If $\CTerm{X}{s}{\vec{x}}{\varphi} \R^p_{\rho} \CTerm{Y}{t}{\vec{y}}{\psi}$ and
$\CTerm{X}{s}{\vec{x}}{\varphi} \R^p_{\rho'} \CTerm{Y'}{t'}{\pvec{y}'}{\psi'}$
then $\CTerm{Y}{t}{\vec{y}}{\psi} \sim \CTerm{Y'}{t'}{\pvec{y}'}{\psi'}$.
\end{restatable}

\begin{proof}
Suppose that
$\CTerm{X}{s}{\vec{x}}{\varphi} \R^p_{\rho,\gamma} \CTerm{Y}{t}{\vec{y}}{\psi}$
and 
$\CTerm{X}{s}{\vec{x}}{\varphi} \R^p_{\rho',\gamma'} \CTerm{Y'}{t'}{\pvec{y}'}{\psi'}$,
where
$\rho\colon \CRu{Z}{\ell}{r}{\pi}$ and $\rho'\colon \CRu{Z'}{\ell'}{r'}{\pi'}$.
Then we have 
(1) $\Dom(\gamma) = \Var(\ell)$,
(2) $s|_p = \ell\gamma$,
(3) $\gamma(x) \in \Val \cup X$ for any $x \in \Var(\ell) \cap Z$, and
(4) $\vDash_\xM (\ECO{\vec{x}}{\varphi}) \Rightarrow (\ECO{\vec{z}}{\pi\gamma})$,
where $\SET{\vec{z}} = \Var(\pi) \setminus \Var(\ell)$,
$t = s[r\gamma]$,
$\psi = \varphi \land \pi\gamma$,
$\SET{\vec{y}} = \Var(\psi) \setminus \Var(t)$, and
$Y = \ExVar(\rho) \cup (X \cap \Var(t))$.
Similarly, we have 
(1') $\Dom(\gamma') = \Var(\ell')$,
(2') $s|_p = \ell'\gamma'$,
(3') $\gamma'(x) \in \Val \cup X$ for any $x \in \Var(\ell') \cap Z'$, and
(4') $\vDash_\xM (\ECO{\vec{x}}{\varphi}) \Rightarrow (\ECO{\pvec{z}'}{\pi'\gamma'})$,
where $\SET{\pvec{z}'} = \Var(\pi') \setminus \Var(\ell')$,
$t' = s[r'\gamma']$,
$\psi' = \varphi \land \pi'\gamma'$,
$\SET{\pvec{y}'} = \Var(\psi') \setminus \Var(t')$, and
$Y' = \ExVar(\rho') \cup (X \cap \Var(t'))$.

Since $\rho,\rho'$ are renamed variants of the same rule,
we have renaming $\sigma\colon \Var(\rho) \to \Var(\rho')$
such that
$\sigma(Z) = Z'$, 
$\sigma(\ell) = \ell'$, 
$\sigma(r) = r'$, and
$\sigma(\pi) = \pi'$.

For each $x \in \Var(r,\pi) \setminus \Var(\ell)$, 
we have by $x \notin \Var(\ell) = \Dom(\gamma)$ that $\gamma(x) = x$,
and by $\sigma(x) \notin \Var(\ell') = \Dom(\gamma')$ that $\gamma'(\sigma(x)) = \sigma(x)$.
Therefore, $\sigma(\gamma(x)) = \sigma(x) = \gamma'(\sigma(x))$ for all $x \in \Var(r,\pi) \setminus \Var(\ell)$.
Suppose $x \in \Var(\ell)$.
Then from $\gamma(x) \in \Var(s)$ and $\Var(s) \cap \Var(\rho) = \varnothing$,
we know $\gamma(x) \notin \Var(\rho)$, and hence $\gamma(x) \notin \Dom(\sigma)$.
Thus, $\sigma(\gamma(x))= \gamma(x)$.
Moreover, since $\ell\gamma = s|_p = \ell'\gamma' = \ell\sigma\gamma'$,
by fixing $x = \ell|_q$ for a position $q$, 
we obtain $\gamma(x) = s|_{pq} = \gamma'(\sigma(x))$.
Thus, $\sigma(\gamma(x))= \gamma(x) = \gamma'(\sigma(x))$ for all $x \in \Var(\ell)$.

Putting all of this together gives
$\sigma(\gamma(x))= \gamma'(\sigma(x))$ for all $x \in \Var(\rho)$.
Hence, we have $r\gamma\sigma = r\sigma\gamma' = r'\gamma'$
and $\pi\gamma\sigma = \pi\sigma\gamma' = \pi'\gamma'$.
Therefore, we have 
$t\sigma = s[r\gamma]_p\sigma = s[r\gamma\sigma]_p = s[r'\gamma']_p = t'$ and 
$\psi\sigma = (\varphi \land \pi\gamma)\sigma = \varphi \land \pi\gamma\sigma 
= \varphi \land \pi'\gamma' = \psi'$.
In particular, it also follows that $\sigma(\Var(t)) = \Var(t\sigma) = \Var(t')$
and $\sigma(\Var(\psi)) = \Var(\psi\sigma) = \Var(\psi')$.
Thus, as $\sigma$ is bijective on $\Var(\rho)$, we obtain
$\sigma(\SET{\vec{y}}) 
= \sigma(\Var(\psi) \setminus \Var(t))
= \sigma(\SET{ x \in \Var(\rho) \mid x \in \Var(\psi) \land x \notin \Var(t) }
= \SET{ \sigma(x) \in \Var(\rho) \mid x \in \Var(\psi) \land x \notin \Var(t) }
= \SET{ \sigma(x) \in \Var(\rho) \mid \sigma(x) \in \Var(\psi') \land \sigma(x) \notin \Var(t') }
= \Var(\psi') \setminus \Var(t') = \SET{\pvec{y}'}$.
From $X \cap \Var(\rho) = \varnothing$, we know that $X \cap \Dom(\sigma) = \varnothing$.
Thus, $\sigma(X) = X$. Therefore,
$\sigma(Y) 
= \sigma(\ExVar(\rho) \cup (X \cap \Var(t)))
= \sigma((\Var(r) \setminus \Var(\ell)) \cup (X \cap \Var(t)))
= (\Var(r\sigma) \setminus \Var(\ell\sigma)) \cup (X \cap \Var(t\sigma))
= \ExVar(\rho') \cup (X \cap \Var(t')) = Y'$.
Since $\sigma$ is a bijective renaming
by 
\Cref{thm: characterization of equivalence using renaming}
we have $\CTerm{Y}{t}{\vec{y}}{\psi} \sim \CTerm{Y'}{t'}{\pvec{y}'}{\psi'}$.
\end{proof}

Uniqueness of reducts implies that for LCTRSs with
a finite number of constrained rewrite rules, there exists 
only a finite number of reducts for any existentially constrained term
(modulo equivalence).
This stands in sharp contrast to the original notion
of rewriting non-quantified terms and represents
a useful property to check convergence of constrained terms.

\begin{restatable}[Commutation of Rewrite Steps and Equivalence for Pattern-General Terms]{theorem}{LemmaIVxi}
\label{lem:commutativity of rewrite steps and equivalence for pattern-general terms}
\label{thm:commutativity of rewrite steps and equivalence for pattern-general terms}
Let 
$\rho$ be a left-linear constrained rewrite rule and
$\CTerm{X}{s}{\vec{x}}{\varphi},
\CTerm{Y}{t}{\vec{y}}{\psi}$ be existentially constrained terms
that are satisfiable and pattern-general.
If $\CTerm{X'}{s'}{\pvec{x}'}{\varphi'} \Lb[\rho] \CTerm{X}{s}{\vec{x}}{\varphi} \sim  \CTerm{Y}{t}{\vec{y}}{\psi}$
then
$\CTerm{X'}{s'}{\pvec{x}'}{\varphi'} \sim \CTerm{Y'}{t'}{\pvec{y}'}{\psi'} \Lb[\rho]  \CTerm{Y}{t}{\vec{y}}{\psi}$
for some $\CTerm{Y'}{t'}{\pvec{y}'}{\psi'}$ (see below). % \NN{(see \Cref{fig:commutation})}.
\begin{center}
\tikzset{snake it/.style={decorate, decoration=snake}}
\begin{tikzpicture}[]
\node (t11) {$\CTerm{Y}{t}{\vec{y}}{\psi}$};
\node[right of = t11, node distance = 1.2in] (t12) {$\CTerm{X}{s}{\vec{x}}{\varphi}$};
\node[below of = t11, node distance = .35in] (t21) {$\CTerm{Y'}{t'}{\pvec{y}'}{\psi'}$};
\node[right of = t21, node distance = 1.2in] (t22) {$\CTerm{X'}{s'}{\pvec{x}'}{\varphi'}$};
\draw[snake it] (t11) -- (t12);
\draw [->,dash pattern=on 2pt off 1pt] (t11) -- (t21);
\draw [->] (t12) -- (t22);
\draw[snake it,dash pattern=on 2pt off 1pt] (t21) -- (t22);
\end{tikzpicture}
\end{center}
\end{restatable}

\begin{proof}
From \Cref{thm:complete characterization of equivalence for most sim-general constrained terms}
it follows that
there is a renaming $\sigma\colon \Var(s,\varphi) \to \Var(t,\psi)$
such that 
(1) $s\sigma = t$, 
(2) $Y = \sigma(X)$, and
(3) $\vDash_\xM ((\ECO{\vec{x}}{\varphi})\sigma \Leftrightarrow (\ECO{\vec{y}}{\psi}))$.

Suppose that $\CTerm{X}{s}{\vec{x}}{\varphi} \R^p_{\rho,\gamma} \CTerm{X'}{s'}{\pvec{x}'}{\varphi'}$,
where $\rho\colon \CRu{Z}{\ell}{r}{\pi}$
such that
$\Var(\rho) \cap \Var(s,\varphi) = \varnothing$.
W.l.o.g.\ we assume $\Var(\rho) \cap \Var(s,\varphi,t,\psi) = \varnothing$.
By definition, we have 
(4) $\Dom(\gamma) = \Var(\ell)$,
(5) $s|_p = \ell\gamma$,
(6) $\gamma(x) \in \Val \cup X$ for any $x \in \Var(\ell) \cap Z$, and
(7) $\vDash_\xM (\ECO{\vec{x}}{\varphi}) \Rightarrow (\ECO{\vec{z}}{\pi\gamma})$,
where $\SET{\vec{z}} = \Var(\pi) \setminus \Var(\ell)$,
$s' = s[r\gamma]$,
$\varphi' = \varphi \land \pi\gamma$,
$\SET{\pvec{x}'} = \Var(\varphi') \setminus \Var(s')$, and
$X' = \ExVar(\rho) \cup (X \cap \Var(s'))$.

We show %that there is 
a rewrite step
$\CTerm{Y}{t}{\vec{y}}{\psi} \R^p_{\rho,\delta}  \CTerm{Y'}{t'}{\pvec{y}'}{\psi'}$,
where $\delta = \SET{ x \mapsto \sigma(\gamma(x)) \mid x \in \Var(\ell) }$.
First we show that
$\CTerm{Y}{t}{\vec{y}}{\psi}$ has a $\rho$-redex at $p$ by $\delta$,
that is,
(4') $\Dom(\delta) = \Var(\ell)$,
(5') $t|_p = \ell\delta$,
(6') $\delta(x) \in \Val \cup Y$ for any $x \in \Var(\ell) \cap Z$, and
(7') $\vDash_\xM (\ECO{\vec{y}}{\psi}) \Rightarrow (\ECO{\vec{z}}{\pi\delta})$,
where $\SET{\vec{z}} = \Var(\pi) \setminus \Var(\ell)$.

(4') is satisfied by definition of $\delta$.
(5') follows as $\ell\delta = \ell\gamma\sigma = (s|_p)\sigma = (s\sigma)|_p = t|_p$, using~(1) and (5).
To show (6'), assume $x \in \Val(\ell) \cap Z$. 
Then $\gamma(x) \in \Val \cup X$ by~(6), and
$\delta(x) = \sigma(\gamma(x)) \in \Val \cup \sigma(X) = \Val \cup Y$, by using~(2).
It remains to show~(7').
From~(7) we have $\vDash_\xM (\ECO{\vec{x}}{\varphi}) \Rightarrow (\ECO{\vec{z}}{\pi\gamma})$,
thus $\vDash_\xM (\ECO{\vec{x}}{\varphi})\sigma \Rightarrow (\ECO{\vec{z}}{\pi\gamma})\sigma$.
From~(3) follows that
$\vDash_\xM (\ECO{\vec{y}}{\psi}) \Rightarrow (\ECO{\vec{z}}{\pi\gamma})\sigma$.
By $\sigma\colon \Var(s,\varphi) \to \Var(t,\psi)$ and $\SET{\vec{z}} \subseteq \Var(\rho)$,
we conclude from $\Var(\rho) \cap \Var(s,\varphi) = \varnothing$ 
that $\SET{\vec{z}} \cap \Var(s,\varphi)  = \varnothing$,
and 
further from $\Var(\rho) \cap \Var(t,\psi) = \varnothing$ 
that $\SET{\vec{z}} \cap \Var(t,\psi) = \varnothing$.
Hence, by definition of $\delta$, 
$(\ECO{\vec{z}}{\pi\gamma})\sigma 
= \ECO{\vec{z}}{\pi\gamma\sigma}
= \ECO{\vec{z}}{\pi\delta}$,
and therefore
$\vDash_\xM (\ECO{\vec{y}}{\psi}) \Rightarrow (\ECO{\vec{z}}{\pi\delta})$ is satisfied.
Let
$t' = s[r\delta]_p$,
$\psi' = \psi \land \pi\delta$,
$\SET{\pvec{y}'} = \Var(\psi') \setminus \Var(t')$, and
$Y' = \ExVar(\rho) \cup (Y \cap \Var(t'))$,
we obtain the rewrite step
$\CTerm{Y}{t}{\vec{y}}{\psi} \R^p_{\rho,\delta}  \CTerm{Y'}{t'}{\pvec{y}'}{\psi'}$.

It remains to show that
$\CTerm{X'}{s'}{\pvec{x}'}{\varphi'} \sim \CTerm{Y'}{t'}{\pvec{y}'}{\psi'}$.
To that end we use \Cref{thm:complete characterizatin of equivalence of constrained terms}.
We show 
(8) $\Pos_{X'\cup\Val}(s') = \Pos_{Y'\cup\Val}(t')~ ( = \SET{p_1,\ldots,p_n})$,
(9) 
$\rho(s'[\,]_{p_1,\ldots,p_n}) = t'[\,]_{p_1,\ldots,p_n}$
for some renaming $\rho\colon \Var(s') \setminus X' \to \Var(t') \setminus Y'$,
(10) for any $i,j \in \SET{1,\ldots,n}$, 
        $\vDash_\xM (\ECO{\vec{x}}{\varphi}) \Rightarrow (s'|_{p_i} = s'|_{p_j})$
        iff 
        $\vDash_\xM (\ECO{\vec{y}}{\psi}) \Rightarrow (t'|_{p_i} = t'|_{p_j})$,
(11) for any $i \in \SET{1,\ldots,n}$ and $v \in \Val$, 
        $\vDash_\xM (\ECO{\vec{x}}{\varphi}) \Rightarrow (s'|_{p_i} = v)$
        iff 
        $\vDash_\xM (\ECO{\vec{y}}{\psi}) \Rightarrow (t'|_{p_i} = v)$,
        and
(12) 
let ${\sim} = {\sim}_{\Pos_{X'\cup \Val}(s')} = {\sim}_{\Pos_{X'\cup \Val}(t')}$
and $\mu_{X'},\mu_{Y'}$ be representative substitutions of
$\CTerm{X'}{s'}{\pvec{x}'}{\varphi'}$ and $\CTerm{Y'}{t'}{\pvec{y}'}{\psi'}$, respectively, 
based on the same representative for
each equivalence class $[p_i]_\sim$ ($1 \le i \le n$), 
and we have 
    $\vDash_\xM (\ECO{\pvec{x}'}{\varphi'})\mu_{X'}\theta|_{\tilde{X'}}
    \Leftrightarrow (\ECO{\pvec{y}'}{\psi'})\mu_{Y'}$
    with a renaming $\theta|_{\tilde{X'}} \colon \tilde{X'} \to \tilde{Y'}$,
    where $\theta =  \SET{ \langle s'|_{p_i}, t'|_{p_i} \rangle \mid 1 \le i \le n }$,
    $\tilde{X'} = \hat{X'} \cap X'$, and $\tilde{Y'} = \hat{Y'}\cap Y'$.
    
We start by showing~(8).
Consequently observe that
\addtocounter{equation}{12}
\begin{equation}\label{eq:13}
s'\sigma = s[r\gamma]_p\sigma = s\sigma[r\gamma\sigma]_p =  t[r\delta]_{p} = t'
\end{equation}
Hence, %also 
$\sigma(X') 
= \sigma(\ExVar(\rho) \cup \sigma(X \cap \Var(s'))
= \sigma(\ExVar(\rho)) \cup \sigma(X \cap \Var(s'))
= \ExVar(\rho) \cup (\sigma(X) \cap \Var(s'\sigma))
= \ExVar(\rho) \cup (Y \cap \Var(t')) = Y'$ follows.
This implies that
$\Pos_{X'\cup\Val}(s') 
= \Pos_{\sigma(X')\cup\Val}(s'\sigma) 
= \Pos_{Y'\cup\Val}(t')$,
which shows~(8). 
In the following let $\Pos_{X'\cup\Val}(s') = \SET{p_1,\ldots,p_n}$.

We proceed to show~(9). 
We let $s' = s'[s_1',\ldots,s_n']_{p_1,\ldots,p_n}$
and $t' = t'[t_1',\ldots,t_n']_{p_1,\ldots,p_n}$,
using the fact~(8).
Then, from~(\ref{eq:13}), it follows that
$s'\sigma[~]_{p_1,\ldots,p_n} = t'[~]_{p_1,\ldots,p_n}$ and
$s_i'\sigma = t_i'$ for each $1 \le i \le n$.
Note here that $s_i' \in \Val \cup X'$ and $t_i' \in \Val \cup Y'$.
Thus, since $\sigma$ is a renaming, we have
$s_i' \in X'$ iff $t_i' \in Y'$,
$s_i' \in \Val$ iff $t_i' \in \Val$, and
$s_i' = s_j'$ iff $t_i' = t_j'$
for $i,j \in \SET{1, \ldots, n}$.
Hence $\SET{s_1',\ldots,s_n'} \cap X' = \SET{t_1',\ldots,t_n'} \cap Y'$.
Since $\sigma(X') = Y'$, it implies that 
$\eta = \SET{ x \mapsto \sigma(x) \mid x \in \Var(s') \setminus X' }$
is a bijection from $\Var(s') \setminus X'$ to $\Var(t') \setminus Y'$
such that $\eta(s'[~]_{p_1,\ldots,p_n}) = t'[~]_{p_1,\ldots,p_n}$.

Before we proceed to the remainder of the proof, let us introduce a useful
temporary notation and show some related properties.
For $\ECO{\vec{x}}{\varphi}$ 
and $\ECO{\pvec{x}'}{\varphi'} = \ECO{\pvec{x}'}{\varphi \land \pi\gamma}$,
we have, by \Cref{lem:bound variables of constraint of reducts},
that $\SET{\vec{x}} \subseteq \SET{\pvec{x}'}$.
Let $\SET{\pvec{x}''} = \SET{\pvec{x}'} \setminus \SET{\vec{x}}$
and let us write $\varphi = \varphi(\vec{x},\pvec{x}'')$.
Also, for $\pi\gamma$, because $\SET{\vec{x}} \cap \Var(\pi\gamma) = \varnothing$
by \Cref{lem:bound variables of constraint of reducts},
we have $\Var(\pi\gamma) \cap \SET{\pvec{x}'} = \SET{\pvec{x}'''} \subseteq \SET{\pvec{x}''}$.
Thus, using the subsequence $\pvec{x}'''$ of $\pvec{x}''$,
we write $\pi\gamma = \pi\gamma(\pvec{x}''')$.
Similarly, for $\ECO{\vec{y}}{\psi}$ 
and $\ECO{\pvec{y}'}{\psi'} = \ECO{\pvec{y}'}{\psi \land \pi\delta}$,
let us write $\psi = \psi(\vec{y},\pvec{y}'')$
and $\pi\delta = \pi\delta(\pvec{y}''')$ where 
$\SET{\pvec{y}''} = \SET{\pvec{y}'}\setminus \SET{\vec{y}}$
and $\SET{\pvec{y}'''} \subseteq \SET{\pvec{y}''}$.
Note that 
$\Var(\pi\gamma) 
= (\bigcup_{x \in \Var(\ell)\cap \Var(\pi)} \Var(x\gamma)) \cup (\Var(\pi) \setminus \Var(\ell))$.
%%%
Moreover, we obtain that
$\sigma(\SET{\pvec{x}'''})
= \sigma(\Var(\pi\gamma) \cap \SET{\pvec{x}'})
= \sigma(
((\bigcup_{x \in \Var(\ell)\cap \Var(\pi)} \Var(x\gamma)) \cup (\Var(\pi) \setminus \Var(\ell)))
\setminus \Var(s'))
= 
((\bigcup_{x \in \Var(\ell)\cap \Var(\pi)} \Var(x\delta)) \cup (\Var(\pi) \setminus \Var(\ell)))
\setminus \Var(t')
= \Var(\pi\delta) \cap \SET{\vec{y}}
= \SET{\pvec{y}'''}$.

From~(3) $\vDash_\xM ((\ECO{\vec{x}}{\varphi})\sigma \Leftrightarrow (\ECO{\vec{y}}{\psi}))$,
we have that, for any valuation $\xi$,
$\vDash_{\xM,\xi} \varphi\sigma(\vec{a},\pvec{x}''\sigma)$ for some $\vec{a} \in \vert \xM \vert^*$
iff 
$\vDash_{\xM,\xi} \psi(\vec{b},\pvec{y}'')$ for some $\vec{b} \in \vert \xM \vert^*$.
Thus, for any valuation $\xi$,
$\vDash_{\xM,\xi} \varphi\sigma(\vec{a},\pvec{x}''\sigma) \land \pi\gamma\sigma(\pvec{x}'''\sigma)$ 
for some $\vec{a} \in \vert \xM \vert^*$
iff 
$\vDash_{\xM,\xi} \psi(\vec{b},\pvec{y}'') \land \pi\delta(\pvec{y}''')$ 
for some $\vec{b} \in \vert \xM \vert^*$,
as we know $\pi\gamma\sigma = \pi\delta$
and $\sigma(\SET{\pvec{x}'''}) = \SET{\pvec{y}'''}$.
In particular, we have
\begin{equation}\label{eq:14}
\vDash_\xM ((\ECO{\pvec{x}'}{\varphi'})\sigma \Leftrightarrow (\ECO{\pvec{y}'}{\psi'}))
\end{equation}

Using~(\ref{eq:14}) and our assumption,  %that $\sigma$ is a bijection, we have 
$\vDash_\xM \ECO{x'}{\varphi'} \Rightarrow (s'|_{p_i} = s'|_{p_j})$
iff 
$\vDash_\xM \ECO{x'}{\varphi'}\sigma \Rightarrow (\sigma(s'|_{p_i}) = \sigma(s'|_{p_j}))$
iff 
$\vDash_\xM \ECO{y'}{\psi'} \Rightarrow (s'\sigma|_{p_i} = s'\sigma|_{p_j})$
iff 
$\vDash_\xM \ECO{y'}{\psi'} \Rightarrow (t'|_{p_i}) = t'|_{p_j})$
for $i,j \in \SET{1, \ldots,  n}$.
Similarly,
$\vDash_\xM \ECO{x'}{\varphi'} \Rightarrow (s'|_{p_i} = v)$
iff 
$\vDash_\xM \ECO{x'}{\varphi'}\sigma \Rightarrow (\sigma(s'|_{p_i}) = v\sigma)$
iff 
$\vDash_\xM \ECO{y'}{\psi'} \Rightarrow (s'\sigma|_{p_i} = v)$
iff 
$\vDash_\xM \ECO{y'}{\psi'} \Rightarrow (t'|_{p_i} = v)$
for $i,j \in \SET{1, \ldots,  n}$.
This shows~(10) and~(11).

Finally, take $\theta = \SET{\langle s'|_{p_i},t'|_{p_i} \rangle \mid 1 \le i \le n}$.
As $\theta|_{\tilde{X'}} \subseteq \sigma$,
from (10), (11) and $s'\sigma = t'$,
it follows that 
$\theta|_{\tilde{X'}} \circ \mu_{X'} = \mu_{Y'} \circ \sigma$.
From (\ref{eq:14}), we have 
$\vDash_\xM (\ECO{\pvec{x}'}{\varphi'})\sigma\mu_{Y'} 
\Leftrightarrow (\ECO{\pvec{y}'}{\psi'})\mu_{Y'}$, 
and thus,
$\vDash_\xM (\ECO{\pvec{x}'}{\varphi'})\mu_{X'}\theta|_{\tilde{X'}}
\Leftrightarrow (\ECO{\pvec{y}'}{\psi'})\mu_{Y'}$. 
As $\mu_{X'},\mu_{Y'}$ are based on the same representative for
each equivalence class $[p_i]_\sim$ ($1 \le i \le n$), 
$\theta|_{\tilde{X'}}$ is a variable renaming from $\tilde{X'}$ to $\tilde{Y'}$.
This shows~(12) and completes the proof.
\end{proof}

In general, the commutation property allows us to
postpone all equivalence translations, depicted by
the following diagram:
\begin{center}
\tikzset{snake it/.style={decorate, decoration=snake}}
\begin{tikzpicture}[]
\node (t11) {$\CTerm{X}{s}{\pvec{x}}{\varphi}$};
\node[right of = t11, node distance = .7in] (t12) {$\cdot$};
\node[below of = t11, node distance = .35in] (t21) {$\CTerm{X'}{s'}{\pvec{x'}}{\varphi'}$};
\node[right of = t21, node distance = .7in] (t22) {$\cdot$};
\node[right of = t22, node distance = .4in] (t23) {$\cdot$};
\node[below of = t21, node distance = .35in] (t31) {$\CTerm{X''}{s''}{\pvec{x}''}{\varphi''}$};
\node[below of = t23, node distance = .35in] (t33) {$\cdot$};
\node[right of = t33, node distance = .4in] (t34) {$\cdot$};
\node[below of = t34, node distance = .35in] (t44) {$\cdot$};
\node[below of = t31, node distance = .35in] (t41) {$\CTerm{X'''}{s'''}{\pvec{x}'''}{\varphi''}$};
\node[right of = t44, node distance = .7in] (t45) {$\CTerm{Y}{t}{\pvec{y}}{\psi}$};
\draw[snake it] (t11) -- (t12);
\draw [->,dash pattern=on 2pt off 1pt] (t11) -- (t21);
\draw [->] (t12) -- (t22);
\draw[snake it,dash pattern=on 2pt off 1pt] (t21) -- (t22);
\draw [snake it]  (t22) -- (t23);
\draw [->,dash pattern=on 2pt off 1pt] (t21) -- (t31);
\draw [->]  (t23) -- (t33);
\draw[snake it,dash pattern=on 2pt off 1pt] (t31) -- (t33);
\draw [snake it] (t33) -- (t34);
\draw [->] (t34) -- (t44);
\draw [->,dash pattern=on 2pt off 1pt] (t31) -- (t41);
\draw[snake it,dash pattern=on 2pt off 1pt] (t41) -- (t44);
\draw [snake it] (t44) -- (t45);
\end{tikzpicture}
\end{center}
Here, the rewrite steps $\CTerm{X}{s}{\pvec{x}}{\varphi} \Rs^* \CTerm{Y}{t}{\pvec{y}}{\psi}$ can be obtained via the rewrite steps
$\CTerm{X}{s}{\pvec{x}}{\varphi} \R \CTerm{X'}{s'}{\pvec{x}'}{\varphi'} \R \CTerm{X''}{s''}{\pvec{x}''}{\varphi''} \R \CTerm{X'''}{s'''}{\pvec{x}'''}{\varphi'''} \sim \CTerm{Y}{t}{\pvec{y}}{\psi}$
if repeated applications of commutation are possible.
However, 
\Cref{thm:commutativity of rewrite steps and equivalence for pattern-general terms}
does not allow repeated applications, because, in general, reducts of most general
rewrite steps do not necessarily lead again to pattern-general constrained terms.
Hence, its restriction 
to pattern-general terms prevents the possibility of repeated applications.

\section{Simulation via Left-Value-Free Rules}
\label{sec:left-value-free rules}

In \Cref{lem:commutativity of rewrite steps and equivalence for pattern-general terms},
we assume that the equivalent existentially constrained terms are pattern-general.
A natural question is, what happens if we drop this condition?

\begin{example}
\label{ex:non-commutation example}
Let $\rho\colon\CRu{\varnothing}{\m{f}(\m{0})}{\m{0}}{\m{true}}$ be a constrained rule.
Consider the equivalence 
$\CTerm{\varnothing}{\m{f}(\m{0})}{}{}
\sim \CTerm{\varnothing}{\m{f}(x)}{}{x = \m{0}}$
and a rewrite step
$\CTerm{\varnothing}{\m{f}(\m{0})}{}{} \R_\rho \CTerm{\varnothing}{\m{0}}{}{}$.
However, 
$\CTerm{\varnothing}{\m{f}(x)}{}{x = \m{0}}$ cannot be rewritten by $\rho$,
because there does not exist a $\gamma$ such that $\m{f}(x) = \m{f}(\m{0})\gamma$.
Thus constrained terms that are not pattern-general,
may not commute with rewrite steps in general.
\end{example}

However, we can recover commutation using a reasonable transformation
on constrained rewrite rules. In this section, we introduce this 
transformation and show that any constrained rewrite rule can be
simulated by using it.

The counterexample of commutation in \Cref{ex:non-commutation example}
uses the rule $\rho\colon\CRu{\varnothing}{\m{f}(\m{0})}{\m{0}}{\m{true}}$,
which is not able to reduce the constrained term
$\CTerm{\varnothing}{\m{f}(x)}{}{x = \m{0}}$.
However, the rule 
$\rho'\colon\CRu{\SET{x}}{\m{f}(x)}{\m{0}}{x = \m{0}}$ has the same effect as $\rho$,
but can reduce the term as follows:
$\CTerm{\varnothing}{\m{f}(x)}{}{x = \m{0}}
\R_{\rho'} \CTerm{\varnothing}{\m{0}}{x}{x = \m{0}}
\sim \CTerm{\varnothing}{\m{0}}{}{}$.
In the following, we show how this can be achieved for arbitrary constrained rewrite rules.
We define left-value-free rewrite rules~\cite{Kop17} in our setting as follows.

\begin{definition}[Left-Value-Free Rewrite Rule]
A left-linear constrained 
rewrite rule $\CRu{Z}{\ell}{r}{\pi}$ is \emph{left-value-free}
if $\Val(\ell) = \varnothing$.
\end{definition}

\begin{definition}[Left-Value-Free Transformation]
Let $\rho\colon \CRu{Z}{\ell}{r}{\pi}$ be a left-linear constrained rule,
$\Pos_{\Val}(\ell) = \SET{p_1,\ldots,p_n}$, and
$\ell|_{p_i} = v_i$ for $1 \le i \le n$.
Taking the fresh variables $x_1,\ldots,x_n$,
we define the \emph{left-value-free transformation} of $\rho$
as $\lvf(\rho) = \CRu{\hat{Z}}{\hat\ell}{r}{\hat\pi}$ 
where $\hat{Z} = Z \cup \SET{x_1,\ldots,x_n}$,
$\hat\ell = \ell[x_1,\ldots,x_n]_{p_1,\ldots,p_n}$,
and $\hat\pi = \pi \land \bigwedge_{i = 1}^{n} (x_i = v_i)$.
\end{definition}

We demonstrate that restricting rewrite rules to left-value-free ones preserves
generality.

\begin{restatable}{theorem}{TheoremIVxv}
\label{thm:preservation of generality by the left-value-freeness}
Let $\rho\colon \CRu{Z}{\ell}{r}{\pi}$ be a left-linear constrained rewrite rule,
$\lvf(\rho) = \hat\rho\colon \CRu{\hat Z}{\hat \ell}{r}{\hat \pi}$,
and 
$\CTerm{X}{s}{\vec{x}}{\varphi}$ an existentially constrained term
such that $\Var(\hat \rho) \cap \Var(s,\varphi) = \varnothing$.
If $\CTerm{X}{s}{\vec{x}}{\varphi} \R^p_{\rho} \CTerm{Y}{t}{\vec{y}}{\psi}$
then there exists a constrained term $\CTerm{Y'}{t'}{\pvec{y}'}{\psi'}$
such that
$\CTerm{X}{s}{\vec{x}}{\varphi} \R^p_{\hat\rho} 
\CTerm{Y'}{t'}{\pvec{y}'}{\psi'}
\sim \CTerm{Y}{t}{\vec{y}}{\psi}$.
\end{restatable}

\begin{proof}
Suppose
$\ell = \ell[v_1,\ldots,v_n]_{p_1,\ldots,p_n}$,
where
$\SET{p_1,\ldots,p_n} = \Pos_{\Val}(\ell)$.
Let $y_1,\ldots,y_n$ are fresh variables,
$\hat Y = \SET{y_1,\ldots,y_n}$,
$\hat{Z} = Z \cup \hat Y$,
$\hat\ell = \ell[y_1,\ldots,y_n]_{p_1,\ldots,p_n}$, and
$\hat\pi = \pi \land \bigwedge_{i = 1}^{n} (y_i = v_i)$.
We obtain 
$\Var(\hat\ell)  \cap \hat Z   
= (\Var(\ell) \cup \hat Y) \cap (Z \cup \hat Y)
= (\Var(\ell) \cap Z) \cup \hat Y$
and
$\Var(\hat\pi) \setminus \Var(\hat\ell)
= (\Var(\pi) \cup \hat Y) \setminus (\Var(\ell) \cup \hat Y)
= \Var(\pi) \setminus \Var(\ell)$.
Suppose $\CTerm{X}{s}{\vec{x}}{\varphi} \R^p_{\rho,\gamma} \CTerm{Y}{t}{\vec{y}}{\psi}$.
Then we have 
(1) $\Dom(\gamma) = \Var(\ell)$,
(2) $s|_p = \ell\gamma$,
(3) $\gamma(x) \in \Val \cup X$ for any $x \in \Var(\ell) \cap Z$, and
(4) $\vDash_\xM (\ECO{\vec{x}}{\varphi}) \Rightarrow (\ECO{\vec{z}}{\pi\gamma})$,
where $\SET{\vec{z}} = \Var(\pi) \setminus \Var(\ell)$,
$t = s[r\gamma]$,
$\psi = \varphi \land \pi\gamma$,
$\SET{\vec{y}} = \Var(\psi) \setminus \Var(t)$, and
$Y = \ExVar(\rho) \cup (X \cap \Var(t))$.

Let $\delta = \gamma \cup \SET{y_i \mapsto v_i \mid 1 \le i \le n}$.
We first show that 
the term $\CTerm{X}{s}{\vec{x}}{\varphi}$ has a $\hat\rho$-redex 
at $p \in \Pos(s)$ using $\delta$.
By~(1), we have 
$\Dom(\delta) 
= \Dom(\gamma) \cup \hat Y
= \Var(\ell) \cup \hat Y$.
From~(2), we have 
$s|_p = \ell\gamma = \ell[v_1,\ldots,v_n]\gamma 
= \ell\gamma [v_1,\ldots,v_n] 
= \ell\delta[\delta(y_1),\ldots,\delta(y_n)] 
= \ell[y_1,\ldots,y_n] \delta
= \hat\ell\delta$.
We proceed to show that $\delta(x) \in \Val \cup X$
for any $x \in \Var(\hat \ell) \cap \hat Z$.
Let $x \in \Var(\hat\ell) \cap \hat{Z}= (\Var(\ell) \cap Z) \cup \SET{y_1,\ldots,y_n}$.
If $x \in \Var(\ell) \cap Z$, then $\delta(x) = \gamma(x)\in \Val \cup X$.
Otherwise, $x \in \hat Y$, and thus 
$\delta(x) = \delta(y_i) = v_i \in \Val \subseteq \Val \cup X$ for some $1 \leqslant i \leqslant n$.
Using the fact that $\Var(\hat\pi) \setminus \Var(\hat\ell) = \Var(\pi) \setminus \Var(\ell)$,
we show that $\vDash_\xM (\ECO{\vec{x}}{\varphi}) \Rightarrow (\ECO{\vec{z}}{\hat\pi\delta})$.
Here we have 
$\hat\pi\gamma 
= (\pi \land \bigwedge_{i = 1}^{n} (y_i = v_i))\gamma
= \pi\gamma \land \bigwedge_{i = 1}^{n} (v_i = v_i)
= \pi\delta \land \bigwedge_{i = 1}^{n} (v_i = v_i)$.
Clearly $\vDash_\xM \pi\gamma \Leftrightarrow \hat\pi\delta$.
Hence, from~(4), we have 
$\vDash_\xM (\ECO{\vec{x}}{\varphi}) \Rightarrow (\ECO{\vec{z}}{\hat\pi\delta})$.
We conclude that 
$\CTerm{X}{s}{\vec{x}}{\varphi}$ has a $\hat\rho$-redex 
at $p \in \Pos(s)$ using $\delta$.

Hence, we obtain %the rewrite step 
$\CTerm{X}{s}{\vec{x}}{\varphi} \R^p_{\hat\rho,\delta} 
\CTerm{Y'}{t'}{\pvec{y}'}{\psi'}$,
where
$t' = s[r\delta]$,
$\psi' = \varphi \land \hat\pi\delta$,
$\SET{\pvec{y}'} = \Var(\psi') \setminus \Var(t')$, and
$Y' = \ExVar(\hat\rho) \cup (X \cap \Var(t'))$.
From $\ExVar(\hat\rho) \cap \hat Y = \varnothing$,
we have $r\delta = r\gamma$ and $\pi\delta = \pi\gamma$.
Thus,~(5) $t' = s[r\delta] = s[r\gamma] = t$.
It also follows from $\psi = \varphi \land \pi\gamma$,
$\psi' = \varphi \land \hat\pi\delta$,
and $\vDash_\xM \pi\gamma \Leftrightarrow \hat\pi\delta$ (shown above),
that~(6) $\vDash_\xM \psi' \Leftrightarrow \psi$.
Using~(5) and 
$\Var(\hat\pi\delta) 
= \Var(\pi\delta \land \bigwedge_{i = 1}^{n} (v_i = v_i)) 
= \Var(\pi\gamma)$, we obtain that,~%
(7) $\SET{\pvec{y}'} 
= \Var(\psi') \setminus \Var(t')
= \Var(\varphi \land \hat\pi\delta) \setminus \Var(t')
= \Var(\varphi \land \pi\gamma) \setminus \Var(t)
= \Var(\psi) \setminus \Var(t)
= \SET{\vec{y}}$.
By $\ExVar(\hat\rho) \cap \hat Y = \varnothing$,
we have $\ExVar(\rho) =  \ExVar(\hat \rho)$.
(8) $Y' 
= \ExVar(\hat\rho) \cup (X \cap \Var(t'))
= \ExVar(\rho) \cup (X \cap \Var(t))
= Y$.
Having~(5)--(8),
it is easy to check from the definition of equivalence
that
$\CTerm{Y'}{t'}{\vec{'y}}{\psi'} \sim \CTerm{Y}{t}{\vec{y}}{\psi}$
holds.
\end{proof}

\begin{example}
\label{ex:non-commutation example II}
Let us revisit \Cref{ex:non-commutation example}.
There we consider the constrained rewrite rule 
$\rho\colon\CRu{\varnothing}{\m{f}(\m{0})}{\m{0}}{\m{true}}$
and the equivalence
$\CTerm{\varnothing}{\m{f}(\m{0})}{}{} 
\sim \CTerm{\varnothing}{\m{f}(x)}{}{x = \m{0}}$.
Let $\hat\rho = \lvf(\rho) 
= \CRu{\SET{y}}{\m{f}(y)}{\m{0}}{\m{true} \land y = \m{0}}$.
Using the rule $\hat\rho$, instead of $\rho$,
yields
$\CTerm{\varnothing}{\m{f}(\m{0})}{}{} 
\R_{\hat\rho} \CTerm{\varnothing}{\m{0}}{}{\m{true} \land \m{0} = \m{0}}$
and 
$\CTerm{\varnothing}{\m{f}(x)}{}{x = \m{0}} 
\R_{\hat\rho} \CTerm{\varnothing}{\m{0}}{x}{\m{true} \land (x = \m{0}) \land (x = \m{0})}$.
In this case, clearly,
$\CTerm{\varnothing}{\m{0}}{}{\m{true} \land \m{0} = \m{0}}
\sim \CTerm{\varnothing}{\m{0}}{x}{\m{true} \land (x = \m{0}) \land (x = \m{0})}$.
\end{example}

In the next section
we show that, if we only consider left-value-free constrained rewrite rules,
then the commutation property holds for all existentially constrained terms 
that are not necessarily pattern-general.
Considering left-value-free rewrite rules
does not weaken the applicability of the commutation property, because
every constrained rewrite rule can be translated to 
a left-value-free rule, while preserving its effect,
using \Cref{thm:preservation of generality by the left-value-freeness}.

\section{General Commutation for Rewrite Steps with Left-Value-Free Rules and Equivalence}
\label{sec:general commutativity left-value-free rules}

By focusing on pattern-general terms, we have shown
in \Cref{lem:commutativity of rewrite steps and equivalence for pattern-general terms}
that rewrite steps commute with the equivalence transformation.
However, this commutation property inhabits issues w.r.t.\ to repeated applications
as explained at the end of \Cref{sec:commutativity of rewriting and equivalence}.

In the following we show that commutation holds without that restriction %.
if the employed rule is left-value-free.
Since any rule has an equivalent left-value-free rule,
by \Cref{thm:preservation of generality by the left-value-freeness},
one can repeatedly apply the commutation property.
This guarantees that the equivalence transformations
can be postponed till the end of the rewrite sequence.
We first introduce a useful notation that is used 
in the following proofs.

\begin{restatable}{definition}{DefinitioinIVxvii}
\label{def:bullet translation}
Let $X$ be a set of variables and $s$ a term.
Suppose $\SET{p_1,\ldots,p_n} = \Pos_{X \cup  \Val}(s)$.
We define $s^{\bullet_{X}} = s[x_1,\ldots,x_n]_{p_1,\ldots,p_n}$
for some pairwise distinct fresh variables $x_1,\ldots,x_n$.
We denote by $u^{\bullet_{X}} = s^{\bullet_{X}}|_p$ each subterm $u = s|_p$ of $s$.
Consider a linear term $\ell$ together with a substitution $\gamma$ such that
$s = \ell\gamma$ with $\Dom(\gamma) = \Var(\ell)$.
We define a substitution $\gamma^{\bullet_{X}}$ 
by $\gamma^{\bullet_{X}}(x) = \gamma(x)^{\bullet_{X}}$.
When no confusion arises, the superscript $\bullet_{X}$ will be abbreviated 
by ${\bullet}$.
\end{restatable}

We now develop several properties of $\bullet$-translation;
all proofs are given in the appendix.

\begin{restatable}{lemma}{LemmaIVxix}
\label{lem:bullet translation and back}
Let $X$ be a set of variables and $s$ a term.
Let $s^{\bullet_{X}} = s[x_1,\ldots,x_n]_{p_1,\ldots,p_n}$
and $\sigma = \{ x_i \mapsto s|_{p_i} \mid 1 \le i \le n \}$.
Then, $u^{\bullet_X}\sigma = u$ holds for any subterm $u$ of $s$.
\end{restatable}

\begin{restatable}{lemma}{LemmaIVxx}
\label{lem:bullet translation and substitution}
Let $X$ be a set of variables, $s$ a term, and $\ell$ a linear term
such that $s|_p = \ell\gamma$ with $\Dom(\gamma) = \Var(\ell)$
for a substitution $\gamma$.
If $\ell$ is value-free 
then $s^{\bullet_{X}}|_p = \ell\gamma^{\bullet_{X}}$.
\end{restatable}

\begin{restatable}{lemma}{LemmaIVxxi}
\label{lem:bullet translation and validity}
Let $X$ be a set of variables and $s$ a term.
Suppose that $\Pos_{X \cup \Val}(s)= \SET{p_1,\ldots,p_n}$
and $s^{\bullet_X} = s[x_1,\ldots,x_n]_{p_1,\ldots,p_n}$.
Then, for any subterm $u$ of $s$, we have
$\vDash_{\xM} (\bigwedge_{i=1}^n (s(p_i) = x_i)) \Rightarrow u = u^{\bullet_X}$.
\end{restatable}

Below we use yet another specific characterization of
equivalence given as follows, whose proof is given in the appendix.

\begin{restatable}{lemma}{LemmaIVxii}
\label{lem:equivalence by mapping}
Let $\CTerm{X}{s}{\vec{x}}{\varphi},\CTerm{Y}{t}{\vec{y}}{\psi}$ be satisfiable existentially constrained terms.
Suppose 
(1) $\sigma\colon V \to \xV  \cup \Val$ with 
$V \subseteq X$ and 
(2) $s\sigma = t$,
(3) 
for any $x,y \in V$,
$\sigma(x) = \sigma(y)$ implies $\vDash_\xM (\ECO{\vec{x}}{\varphi}) \Rightarrow x = y$.
(4)
$X \setminus V = Y \setminus \sigma(V)$ and
(5)
$\sigma(V) \cap \xV \subseteq Y$,
(6) 
$\FVar((\ECO{\vec{x}}{\varphi})\sigma) = \FVar(\ECO{\vec{y}}{\psi})$
and 
$\vDash_\xM (\ECO{\vec{x}}{\varphi})\sigma \Leftrightarrow (\ECO{\vec{y}}{\psi})$.
Then 
$\CTerm{X}{s}{\vec{x}}{\varphi} \sim \CTerm{Y}{t}{\vec{y}}{\psi}$.
\end{restatable}

We show a key lemma for the theorem that follows.

\begin{restatable}{lemma}{LemmaIVxxii}
\label{lem:equivalence to pattern-general form by left-value-free rules}
Let $\rho$ be a left-value-free constrained rewrite rule.
If $\CTerm{X}{s}{\vec{x}}{\varphi} \R_\rho \CTerm{Y}{t}{\vec{y}}{\psi}$
then $\PG(\CTerm{X}{s}{\vec{x}}{\varphi}) 
\R_\rho \CTerm{Y'}{t'}{\pvec{y}'}{\psi'} \sim  \CTerm{Y}{t}{\vec{y}}{\psi}$
for some $\CTerm{Y'}{t'}{\pvec{y}'}{\psi'}$.
\end{restatable}

\begin{proof}
Let $p \in \Pos(s)$ and 
$\rho\colon \CRu{Z}{\ell}{r}{\pi}$ be a left-linear constrained rewrite rule
such that $\Val(\ell) = \varnothing$ and
$\Var(\rho) \cap \Var(\varphi,s) = \varnothing$.
Suppose $\CTerm{X}{s}{\vec{x}}{\varphi} \R^p_{\rho,\gamma} \CTerm{Y}{t}{\vec{y}}{\psi}$.
Then we have~%
(1) $\Dom(\gamma) = \Var(\ell)$,~%
(2) $s|_p = \ell\gamma$,~%
(3) $\gamma(x) \in \Val \cup X$ for any $x \in \Var(\ell) \cap Z$, and~%
(4) $\vDash_\xM (\ECO{\vec{x}}{\varphi}) \Rightarrow (\ECO{\vec{z}}{\pi\gamma})$,
where $\SET{\vec{z}} = \Var(\pi) \setminus \Var(\ell)$,
$t = s[r\gamma]$,
$\psi = \varphi \land \pi\gamma$,
$\SET{\vec{y}} = \Var(\psi) \setminus \Var(t)$, and
$Y = \ExVar(\rho) \cup (X \cap \Var(t))$.

Using the notation of \Cref{def:bullet translation},
define 
$\PG(\CTerm{X}{s}{\vec{x}}{\varphi}) = \CTerm{X'}{s^{\bullet_{X}}}{\pvec{x}'}{\varphi'}$,
i.e.\
$s^{\bullet_{X}} = s[x_1,\ldots,x_n]_{p_1,\ldots,p_n}$
for set of positions $\SET{p_1,\ldots,p_n} = \Pos_{X\cup\Val}(s)$,
some pairwise distinct fresh variables $x_1,\ldots,x_n$,
$X' = \SET{x_1,\ldots,x_n}$
and $\SET{\pvec{x}'} = \SET{\vec{x}} \cup X$
and 
$\varphi' = (\varphi \land \bigwedge_{i=1}^n (s|_{p_i} = x_i))$.

We first show 
$\CTerm{X'}{s^\bullet}{\pvec{x}'}{\varphi'}$
has a $\rho$-redex at $p$ by $\gamma^\bullet$.
Clearly,
(1') $\Dom(\gamma^\bullet) = \Dom(\gamma) = \Var(\ell)$.
By \Cref{lem:bullet translation and substitution},
we have~(2') $s^\bullet|_p = \ell\gamma^\bullet$.
Suppose $x \in \Var(\ell) \cap Z$.
Then $\gamma(x) \in \Val \cup X$,
which implies $\gamma(x) = s|_{p_i}$ for some $1 \le i \le n$.
Then, 
$\gamma^\bullet(x) 
= \gamma(x)^\bullet
= (s|_{p_i})^\bullet
= s^\bullet|_{p_i}
= (s[x_1,\ldots,x_n]_{p_1,\ldots,p_n})|_{p_i}
= x_i \in X' \subseteq \Val \cup X'$.
Thus,~(3') 
$\gamma^\bullet(x) \in \Val \cup X'$ for any $x \in \Var(\ell) \cap Z$
follows.
It remains to show~%
(4') $\vDash_\xM (\ECO{\pvec{x}'}{\varphi'}) \Rightarrow (\ECO{\vec{z}}{\pi\gamma^\bullet})$.

For this, 
let $\vDash_{\xM,\xi} \ECO{\pvec{x}'}{\varphi'}$
for a valuation $\xi$.
This means,
$\vDash_{\xM,\xi} \ECO{\vec{x},\pvec{x}''}{\varphi \land \bigwedge_{i=1}^n (s|_{p_i} = x_i))}$,
where $\SET{\pvec{x}''}  = X$.
By \Cref{lem:bullet translation and validity},
$\vDash_{\xM} 
\bigwedge_{i=1}^n (s(p_i) = x_i)) \Rightarrow (\gamma(x) = \gamma(x)^\bullet)$
holds for any $x \in \xV$.
For, it trivially follows if $x \notin \Dom(\gamma)$,
and otherwise, $\gamma(x)$ is a subterm of $s$,
and thus \Cref{lem:bullet translation and validity} applies.
Note that $s(p_i) = s|_{p_i}$ for all $1 \leqslant i \leqslant n$.
This implies that 
$\vDash_{\xM,\xi} \ECO{\vec{x},\pvec{x}''}{\varphi 
\land (\bigwedge_{i=1}^n (s|_{p_i} = x_i))
\land (\bigwedge_{x \in \Var(\pi)}(\gamma(x) = \gamma(x)^\bullet))}$.
Hence, we obtain
$\vDash_{\xM,\xi} \ECO{\vec{x},\pvec{x}''}{\varphi \land 
(\bigwedge_{i=1}^n (s|_{p_i} = x_i)) \land (\pi\gamma \Leftrightarrow \pi\gamma^\bullet)}$.
Now, by definition,
there exists a sequence $\vec{v}\in \Val^*$ %of values 
such that
$\vDash_{\xM,\xi} (\ECO{\vec{x}}{\varphi \land 
(\bigwedge_{i=1}^n (s|_{p_i} = x_i)) \land (\pi\gamma \Leftrightarrow \pi\gamma^\bullet)})\kappa$,
where $\kappa = \SET{\pvec{x}'' \mapsto \vec{v}}$.
By $\Dom(\kappa) = \SET{\pvec{x}''} = X$, it follows that
$\vDash_{\xM,\xi} \ECO{\vec{x}}{\varphi\kappa \land 
(\bigwedge_{i=1}^n (s|_{p_i}\kappa = x_i)) \land (\pi\gamma\kappa = \pi\gamma^\bullet)}$.
In particular, 
$\vDash_{\xM,\xi} \ECO{\vec{x}}{\varphi\kappa}$ holds and then
by $\Dom(\kappa) \cap \SET{\vec{x}} = \varnothing$,
we have $\vDash_{\xM,\xi} (\ECO{\vec{x}}{\varphi})\kappa$,
i.e.,  $\vDash_{\xM,\xi \circ \kappa} \ECO{\vec{x}}{\varphi}$.
Then, by~(4), it follows $\vDash_{\xM,\xi \circ \kappa} \ECO{\vec{z}}{\pi\gamma}$,
i.e., $\vDash_{\xM,\xi} (\ECO{\vec{z}}{\pi\gamma})\kappa$.
As $\SET{\vec{z}} \cap X = \varnothing$,
this implies $\vDash_{\xM,\xi} \ECO{\vec{z}}{\pi\gamma\kappa}$.
Moreover, by $\SET{\vec{x}} \cap \Var(\pi\gamma\kappa,\pi\gamma^\bullet) = \varnothing$,  
from $\vDash_{\xM,\xi} \ECO{\vec{x}}{(\pi\gamma\kappa \Leftrightarrow \pi\gamma^\bullet)}$
it follows that $\vDash_{\xM,\xi} (\pi\gamma\kappa \Leftrightarrow \pi\gamma^\bullet)$.
Thus, 
we obtain $\vDash_{\xM,\xi} \ECO{\vec{z}}{\pi\gamma^\bullet}$
and have proven~(4').

Hence, we obtain the rewrite step
$\CTerm{X'}{s^\bullet}{\pvec{x}'}{\varphi'}
\R^p_{\rho,\gamma^\bullet} \CTerm{Y'}{t'}{\pvec{y}'}{\psi'}$,
where
$t' = s^\bullet[r\gamma^\bullet]$,
$\psi' = \varphi' \land \pi\gamma^\bullet$,
$\SET{\pvec{y}'} = \Var(\psi') \setminus \Var(t')$, and
$Y' = \ExVar(\rho) \cup (X' \cap \Var(t'))$.

It remains to show 
$\CTerm{Y'}{t'}{\pvec{y}'}{\psi'} \sim \CTerm{Y}{t}{\vec{y}}{\psi}$.
For this, we use \Cref{lem:equivalence by mapping}
by taking
$\sigma = \SET{ x_i \mapsto s|_{p_i} \mid 1 \le i \le n, x_i \in Y' }$
and $V = \SET{ x_i \mid 1 \le i \le n, x_i \in Y' }$.
We refer to the conditions~(1)--(6) of \Cref{lem:equivalence by mapping} by~(C1)--(C6).
Clearly, (C1) $\sigma\colon V \to \xV  \cup \Val$ with $V \subseteq Y'$.
(C2) holds as $t'\sigma 
= s^\bullet[r\gamma^\bullet]\sigma 
= s^\bullet\sigma [r\gamma^\bullet\sigma]
= s[r\gamma] = t$,
using \Cref{lem:bullet translation and back}.
Let $x_i,x_j \in Y'$ and $\sigma(x_i) = \sigma(x_j)$.
Then, $s|_{p_i} = s|_{p_j}$ and hence $\vDash_{\xM} \psi' \Rightarrow (s|_{p_i} = s|_{p_j})$.
By $x_i,x_j \in Y'$, we know that $x_i,x_j \notin \SET{\pvec{y}'}$,
and hence it follows $\vDash_{\xM} \ECO{\pvec{y}'}{\psi'} \Rightarrow (x_i = x_j)$.
From this we have~(C3).
Note that $Y' = \ExVar(\rho) \cup (X' \cap \Var(t'))
= \ExVar(\rho) \cup (\{ x_1,\ldots,x_n \} \cap \Var(t'))
$ and hence $V = \{ x_1,\ldots,x_n \} \cap \Var(t')$
and that $Y = \ExVar(\rho) \cup (\{ s|_{p_1},\ldots, s|_{p_n} \} 
\cap X \cap \Var(t))$.
Since $Y'\setminus V = \ExVar(\rho) = Y \setminus \sigma(V)$,
we have~(C4).
For~(C5), let $x \in Y'$ such that $\sigma(x) \in \xV$.
If $x \in V$ then $\sigma(x) \in Y$.
Otherwise $\sigma(x) = x$, and by $s\sigma = t$,
we know $x \in \Var(t)$. Hence $x \in \ExVar(\rho) \subseteq Y$.
Thus,~(C5) holds.
In order to show (C6), note first that 
$\FVar(\ECO{\pvec{y}'}{\psi'}) = \Var(\psi') \cap \Var(t')$ and
$\FVar(\ECO{\vec{y}}{\psi}) = \Var(\psi) \cap \Var(t)$.
Note that $\psi'\sigma = (\varphi' \land \pi\gamma^\bullet)\sigma 
= (\varphi \land \bigwedge_i (s|_{p_i} = x_i) \land \pi\gamma^\bullet)\sigma
= (\varphi \land \bigwedge_i (s|_{p_i} = s|_{p_i}) \land \pi\gamma)$
and $\psi = (\varphi \land \pi\gamma)$,
and hence $\Var(\psi'\sigma) =\Var(\psi)$. 
By $t'\sigma = t$
it follows
$\FVar((\ECO{\pvec{y}'}{\psi'})\sigma)
= \sigma(\FVar(\ECO{\pvec{y}'}{\psi'})) 
= \FVar(\ECO{\vec{y}}{\psi})$.
Finally, 
\[
\begin{array}{@{}l@{\>}c@{\>}l@{}}
\lefteqn{\vDash_\xM (\ECO{\pvec{y}'}{\psi'})\sigma}\\
&\iff& \vDash_\xM 
(\ECO{\pvec{y}'}{ (\varphi \land (\bigwedge_{i=1}^n (s|_{p_i} = x_i))) \land \pi\gamma^\bullet})\sigma\\
&\iff& \vDash_\xM 
\ECO{\pvec{y}'}{ (\varphi \land (\bigwedge_{i=1}^n (s|_{p_i} = x_i\sigma))) \land \pi\gamma^\bullet\sigma}\\
&\iff& \vDash_\xM \ECO{\pvec{y}'}{(\varphi \land (\bigwedge_{i=1}^n (s|_{p_i} = s|_{p_i})) \land \pi\gamma}\\
&\iff& \vDash_\xM \ECO{\pvec{y}'}{ (\psi \land (\bigwedge_{i=1}^n (s|_{p_i} = s|_{p_i}))}\\
&\iff& \vDash_\xM \ECO{\pvec{y}'}{\psi}
\end{array}
\]
Thus, we have $\vDash_\xM (\ECO{\pvec{y}'}{\psi'})\sigma
 \Leftrightarrow (\ECO{\pvec{y}'}{\psi})$.
Since we can eliminate bound variables which do not appear in the constraint,
we conclude $\vDash_\xM (\ECO{\pvec{y}'}{\psi'})\sigma
 \Leftrightarrow (\ECO{\vec{y}}{\psi})$.
\end{proof}

Finally, we are able to prove the main result.

\begin{restatable}[Commutation of Rewrite Steps and Equivalence by Left-Value-Free Rules]{theorem}{TheoremIVxxiv}
\label{thm:commutativity of rewrite steps and equivalence by left-value-free rules}
Let $\rho$ be a left-value-free constrained rewrite rule,
and $\CTerm{X}{s}{\vec{x}}{\varphi},
\CTerm{Y}{t}{\vec{y}}{\psi}$ be satisfiable existentially constrained terms.
If $\CTerm{X'}{s'}{\pvec{x}'}{\varphi'} \gets_\rho \CTerm{X}{s}{\vec{x}}{\varphi} \sim  \CTerm{Y}{t}{\vec{y}}{\psi}$,
then
we have $\CTerm{X'}{s'}{\pvec{x}'}{\varphi'} \sim \CTerm{Y'}{t'}{\pvec{y}'}{\psi'} \gets_\rho  \CTerm{Y}{t}{\vec{y}}{\psi}$
for some $\CTerm{Y'}{t'}{\pvec{y}'}{\psi'}$.
\end{restatable}

\begin{proof}
Suppose that
$
\CTerm{X'}{s'}{\pvec{x}'}{\varphi'} 
\gets_\rho 
\CTerm{X}{s}{\vec{x}}{\varphi} 
\sim
\CTerm{Y}{t}{\vec{y}}{\psi}
$.
By \Cref{lem:equivalence to pattern-general form by left-value-free rules},
there exists an existentially constrained term $\CTerm{X''}{s''}{\pvec{x}''}{\varphi''}$
such that 
$
\CTerm{X'}{s'}{\pvec{x}'}{\varphi'} 
\sim
\CTerm{X''}{s''}{\pvec{x}''}{\varphi''} 
\gets_\rho 
\PG(\CTerm{X}{s}{\vec{x}}{\varphi})
$.
By
$\CTerm{X}{s}{\vec{x}}{\varphi}
\sim 
\CTerm{Y}{t}{\vec{y}}{\psi}
$,
$
\CTerm{X''}{s''}{\pvec{x}''}{\varphi''} 
\gets_\rho 
\PG(\CTerm{X}{s}{\vec{x}}{\varphi})
\sim 
\PG(\CTerm{Y}{t}{\vec{y}}{\psi})
$.
Hence, by \Cref{lem:commutativity of rewrite steps and equivalence for pattern-general terms},
$
\CTerm{X''}{s''}{\pvec{x}''}{\varphi''}
\sim 
\CTerm{Y'}{t'}{\pvec{y}'}{\psi'}
\gets_\rho 
\PG(\CTerm{Y}{t}{\vec{y}}{\psi})
$
for some existentially constrained term $\CTerm{Y'}{t'}{\pvec{y}'}{\psi'}$.
Hence, 
$
\CTerm{X'}{s'}{\pvec{x}'}{\varphi'} 
\sim
\CTerm{Y'}{t'}{\pvec{y}'}{\psi'}
\gets_\rho 
\PG(\CTerm{Y}{t}{\vec{y}}{\psi})
$.
\end{proof}

\Cref{thm:commutativity of rewrite steps and equivalence by left-value-free rules}
enables us to defer the equivalence transformations
until after the application of rewrite rules.

\begin{restatable}{corollary}{TheoremIVxxv}
\label{cor:commutativity of multi rewrite steps and equivalence by left-value-free rules}
Let $\xR$ be an LCTRS consisting of left-value-free constrained rewrite rules.
If $\CTerm{X}{s}{\vec{x}}{\varphi} \Rs_\xR \CTerm{Y}{t}{\vec{y}}{\psi}$,
then
$\CTerm{X}{s}{\vec{x}}{\varphi} \to^*_\xR \CTerm{X'}{s'}{\pvec{x}'}{\varphi'} 
\sim \CTerm{Y}{t}{\vec{y}}{\psi}$ for some $\CTerm{X'}{s'}{\pvec{x}'}{\varphi'}$.
\end{restatable}

\begin{proof}[Proof Sketch]
The claim can be proved by induction on $n$ to show that
if we have $\CTerm{X}{s}{\vec{x}}{\varphi}
(\sim  \cdot \to_\xR)^n \cdot \sim \CTerm{Y}{t}{\vec{y}}{\psi}$,
then there exists $\CTerm{X'}{s'}{\pvec{x}'}{\varphi'}$
such that 
$\CTerm{X}{s}{\vec{x}}{\varphi} \to^*_\xR \CTerm{X'}{s'}{\pvec{x}'}{\varphi'} 
\sim \CTerm{Y}{t}{\vec{y}}{\psi}$.
\end{proof}

\section{Related Work}
\label{sec:Related Work}
As mentioned earlier, tools working with LCTRSs, like \ctrl{} or \crest{}, 
do already implement (mitigated) versions of the results in this paper. 
Implementation details w.r.t.\ \crest{} are discussed in~\cite{SM23,SM25}. 
This includes transformations like \cite[Definition~12]{SM23} that produces
left-value-free rules, or to check joinability of $s\approx t~\CO{\varphi}$
\crest{} does not compute any intermediate equivalences but checks triviality by
using \cite[Lemma~3]{SM23} at the end of the rewrite sequence. 
\ctrl~\cite{KN15} makes
constrained rewrite rules left-value-free in advance~\cite{FKN17tocl,Kop17}, and
thus, equivalence transformations are not used in applying constrained rewrite
rules to constrained terms. 
\crisys~\cite{KNM25jlamp} does not modify the given
LCTRSs but uses an extension of matching during rewriting constrained terms, which
can be seen as an alternative to equivalence transformations before rewrite steps: 
When matching a term $\ell$ with $s$ in a constrained term
$\CTerm{}{s}{}{\varphi}$, a value $v$ can match a logical variable $x$ if $x =
v$ is guaranteed by $\varphi$, i.e., $\vDash_\xM \varphi \Rightarrow (x = v)$.
However, the results of this paper do show promising evidence that tools
incorporating these approximations do not loose much power. 

Regarding rewriting of constrained terms, the \emph{symbolic rewriting
module SMT} has been proposed in~\cite{RMM17}. The rewriting formalism
there is different from constrained rewriting in this paper. We use the
original definition of LCTRSs and constrained rewriting on them with the
only difference on the presentation to ease its analysis. In particular,
constrained rewriting reduces a constrained term to another one such that every
instance in the former has its reduct in the latter. On the other hand, symbolic
rewriting in~\cite{RMM17} does not guarantee this, hence this is similar to
narrowing of constrained terms. 

\section{Conclusion}
\label{sec:conclusion}

In this paper, we have revisited the formalism of constrained rewriting in LCTRSs.
We have introduced the new notion 
of most general rewriting on 
existentially constrained terms
for left-linear LCTRSs. 
It was discussed in which way our new formalism of most general rewriting 
extracts the so called ``most general'' part of the corresponding original formalism.
We have shown the uniqueness of reducts
for our formalism of constrained rewriting, and 
commutation between rewrite steps and equivalence
for pattern-general constrained terms.
Then, by using the notion of non-left-value-free rewrite rules
we showed
that left-value-free 
rewrite rules can simulate non-left-value-free rules.
Finally, we did recover
general commutation between rewrite steps with left-value-free rules
and equivalent transformations.

Because of the expected complications, our focus in this paper was on rewriting
with left-linear rules. 
To deal with non-left-linear rules, we need to extend
the matching mechanism underlying rewrite steps. However, this
complicates the definition of rewrite steps drastically. 
This extension remains as future work.

\begin{acks}
This work was partially supported by
FWF (Austrian Science Fund) project I~5943-N
and 
JSPS KAKENHI Grant Numbers JP24K02900 and JP24K14817.
\end{acks}

\bibliographystyle{ACM-Reference-Format}
\bibliography{biblio}

@string{proc = "Proceedings of the "}

@string{cade = " CADE"}

@string{IJCAR = " IJCAR"}

@string{lipics = "LIPIcs"}

@string{fscd = " FSCD"}

@string{lopstr = " LOPSTR"}

@string{lpar = " LPAR"}

@string{vstte= " VSTTE"}

@string{aplas= " APLAS"}

@string{frocos= " FroCoS"}

@string{tacas= " TACAS"}

@string{lnai = "Lecture Notes in Artificial Intelligence"}

@string{lncs = "Lecture Notes in Computer Science"}

@book{BN98,
  title = {Term Rewriting and All That},
  author = {Franz Baader and Tobias Nipkow},
  year = {1998},
  doi = {10.1145/505863.505888},
  publisher = {Cambridge University Press},
  xaddress = {},
}

@book{O02,
  title = {Advanced Topics in Term Rewriting},
  author = {Enno Ohlebusch},
  year = {2002},
  doi = {10.1007/978-1-4757-3661-8},
  publisher = {Springer New York},
  address = {New York},
}

@article{TSNA25LOPSTR-arxiv,
 author    = "Kanta Takahata and Jonas Sch{\"o}pf and Naoki Nishida and Takahito Aoto",
 title     = "Characterizing Equivalence of Logically Constrained Terms via Existentially Constrained Terms",
 journal   = "CoRR",
 volume    = "abs/2505.21986",
 year      = 2025,
 pages     = {1-29},
 doi       = "10.48550/arXiv.2505.21986",
 xnote     = "Full version of a paper submitted to LOPSTR 2025",
}

@article{FKN17tocl,
 author = {Carsten Fuhs and Cynthia Kop and Naoki Nishida},
 title = {Verifying Procedural Programs via Constrained Rewriting Induction},
 journal = {ACM Transactions on Computational Logic},
 volume = {18},
 number = {2},
 pages = {14:1--14:50},
 year = {2017},
 doi = {10.1145/3060143},
}

@article{KNM25jlamp,
 author    = {Misaki Kojima and Naoki Nishida and Yutaka Matsubara},
 title     = {Transforming concurrent programs with semaphores into logically constrained term rewrite systems},
 journal   = {Journal of Logical and Algebraic Methods in Programming},
 volume    = {143},
 number    = "",
 pages     = {1-23},
 year      = {2025},
 doi       = {10.1016/j.jlamp.2024.101033},
}

@article{K16,
 author    = "Cynthia Kop",
 title     = "Termination of {LCTRSs}",
 journal   = "CoRR",
 volume    = "abs/1601.03206",
 year      = 2016,
 doi       = "10.48550/ARXIV.1601.03206",
 pages     = "1-5",
}

@inproceedings{SM25,
 author    = "Jonas Sch{\"o}pf and Aart Middeldorp",
 title     = {Automated Analysis of Logically Constrained Rewrite Systems using crest},
 booktitle = proc # "31st" # tacas,
 series    = lncs,
 volume    = 15696,
 pages     = "124--144",
 year      = 2025,
 publisher = "Springer Nature Switzerland",
 address   = "Cham",
 doi       = "10.1007/978-3-031-90643-5_7",
 xnote      = "\doi{10.1007/978-3-031-90643-5_7}"
}

@inproceedings{WM21,
 author    = "Sarah Winkler and Georg Moser",
 editor    = "Maribel Fern{\'a}ndez",
 title     = "Runtime Complexity Analysis of Logically Constrained
              Rewriting",
 booktitle = proc # "30th" # lopstr,
 series    = lncs,
 volume    = 12561,
 pages     = "37--55",
 year      = 2021,
 publisher = "Springer International Publishing",
 address   = "Cham",
 doi       = "10.1007/978-3-030-68446-4_2",
}

@inproceedings{SM23,
 author    = "Jonas Sch{\"o}pf and Aart Middeldorp",
 editor    = "Brigitte Pientka and Cesare Tinelli",
 title     = "Confluence Criteria for Logically Constrained Rewrite
              Systems",
 booktitle = proc # "29th" # cade,
 series    = lnai,
 volume    = "14132",
 pages     = "474--490",
 year      = 2023,
 publisher = "Springer Nature Switzerland",
 address   = "Cham",
 doi       = "10.1007/978-3-031-38499-8_27"
}

@inproceedings{WM18,
 author    = "Sarah Winkler and Aart Middeldorp",
 editor    = "H\'{e}l\`{e}ne Kirchner",
 title     = "Completion for Logically Constrained Rewriting",
 booktitle = proc # "3rd" # fscd,
 series    = lipics,
 volume    = 108,
 pages     = "30:1--30:18",
 year      = 2018,
 publisher = {Schloss Dagstuhl -- Leibniz-Zentrum f{\"u}r Informatik},
 address   = {Dagstuhl, Germany},
 doi       = "10.4230/LIPIcs.FSCD.2018.30"
}

@inproceedings{KN15,
 author = {Cynthia Kop and Naoki Nishida},
 editor = {Martin Davis and Ansgar Fehnker and Annabelle McIver and Andrei Voronkov},
 title = {{C}ons{T}rained {R}ewriting too{L}},
 booktitle = proc # "20th" # lpar,
 series = lncs,
 volume = {9450},
 pages = {549--557},
 year = {2015},
 publisher = "Springer Berlin Heidelberg",
 address = "Berlin, Heidelberg",
 doi = {10.1007/978-3-662-48899-7_38},
}

@inproceedings{KN14aplas,
 author = {Cynthia Kop and Naoki Nishida},
 editor = "Jacques Garrigue",
 title = {Automatic Constrained Rewriting Induction towards Verifying Procedural Programs},
 booktitle = proc # "12th" # aplas,
 series = lncs,
 volume = {8858},
 pages = {334--353},
 year = {2014},
 publisher = "Springer International Publishing",
 address = "Cham",
 doi = {10.1007/978-3-319-12736-1_18},
}

@inproceedings{KN13frocos,
 author = {Cynthia Kop and Naoki Nishida},
 editor = {Pascal Fontaine and Christophe Ringeissen and Renate A. Schmidt},
 title = {Term Rewriting with Logical Constraints},
 booktitle = proc # "9th" # frocos,
 series = lncs,
 volume = {8152},
 pages = {343--358},
 year = {2013},
 publisher = "Springer Berlin Heidelberg",
 address = "Berlin, Heidelberg",
 doi = {10.1007/978-3-642-40885-4_24},
}

@inproceedings{SMM24,
 author    = "Jonas Sch{\"o}pf and  Fabian Mitterwallner and Aart Middeldorp",
 editor    = "Christoph Benzmüller and Marijn J.H. Heule and Renate A. Schmidt",
 title     = "Confluence of Logically Constrained Rewrite Systems Revisited",
 booktitle = proc # "12th" # ijcar,
 series    = lnai,
 volume    = "14740",
 pages     = "298--316",
 year      = 2024,
 publisher = "Springer Nature Switzerland",
 address   = "Cham",
 doi       = "10.1007/978-3-031-63501-4_16"
}

@inproceedings{NW18,
 author    = "Naoki Nishida and Sarah Winkler",
 editor    = "Ruzica Piskac and Philipp R{\"u}mmer",
 title     = "Loop Detection by Logically Constrained Term Rewriting",
 booktitle = proc # "10th" # vstte,
 series    = lncs,
 volume    = 11294,
 pages     = "309--321",
 year      = 2018,
 publisher = "Springer International Publishing",
 address   = "Cham",
 doi       = "10.1007/978-3-030-03592-1_18"
}

@inproceedings{ANS24,
 author    = "Takahito Aoto and Naoki Nishida and Jonas Sch{\"o}pf",
 editor    = "Jakob Rehof",
 title     = "Equational Theories and Validity for Logically Constrained
              Term Rewriting",
 booktitle = proc # "9th" # fscd,
 series    = lipics,
 volume    = 299,
 pages     = "31:1--31:21",
 year      = 2024,
 publisher = {Schloss Dagstuhl -- Leibniz-Zentrum f{\"u}r Informatik},
 address   = {Dagstuhl, Germany},
 doi       = "10.4230/LIPIcs.FSCD.2024.31"
}

@article{Kop17,
  author       = {Cynthia Kop},
  title        = {Quasi-reductivity of Logically Constrained Term Rewriting Systems},
  journal      = {CoRR},
  volume       = {abs/1702.02397},
  pages        = {1-8},
  year         = {2017},
  doi          = {10.48550/arXiv.1702.02397},
  xurl         = {http://arxiv.org/abs/1702.02397},
  xeprinttype  = {arXiv},
  xeprint      = {1702.02397},
}

@article{RMM17,
  author       = {Camilo Rocha and
                  Jos{\'{e}} Meseguer and
                  C{\'{e}}sar A. Mu{\~{n}}oz},
  title        = {Rewriting modulo {SMT} and open system analysis},
  journal      = {Journal of Logic and Algebraic Programming},
  volume       = {86},
  number       = {1},
  pages        = {269-297},
  year         = {2017},
  doi          = {10.1016/J.JLAMP.2016.10.001},
}

\appendix

\section{Omitted Proofs}

\LemmaIVvii*

\begin{proof}
Let $\CTerm{X}{s}{\vec{x}}{\varphi} \R^p_{\rho,\gamma} \CTerm{Y}{t}{\vec{y}}{\psi}$.
Then we have 
(1) $\Dom(\gamma) = \Var(\ell)$,
(2) $s|_p = \ell\gamma$,
(3) $\gamma(x) \in \Val \cup X$ for any $x \in \Var(\ell) \cap Z$, and
(4) $\vDash_\xM (\ECO{\vec{x}}{\varphi}) \Rightarrow (\ECO{\vec{z}}{\pi\gamma})$,
where $\SET{\vec{z}} = \Var(\pi) \setminus \Var(\ell)$,
$t = s[r\gamma]_{p}$,
$\psi = \varphi \land \pi\gamma$,
$\SET{\vec{y}} = \Var(\psi) \setminus \Var(t)$, and
$Y = \ExVar(\rho) \cup (X \cap \Var(t))$.
We have
$\FVar(\ECO{\vec{y}}{\psi}) 
= \Var(\psi) \setminus \SET{\vec{y}} 
= \Var(\psi) \setminus (\Var(\psi) \setminus \Var(t))
= \Var(\psi) \cap \Var(t)
= \Var(\varphi \land \pi\gamma) \cap \Var(t)$.
We show that $x \in \FVar(\ECO{\vec{y}}{\psi})$
implies $x \in \ExVar(\rho) \cup (X \cap \Var(t))$,
and distinguish the cases
whether $x \in \Var(\varphi) \cap \Var(t)$
or $x \in \Var(\pi\gamma) \cap \Var(t)$.
\begin{itemize}
\item 
Case $x \in \Var(\varphi) \cap \Var(t)$.
Because $\Var(\rho) \cap \Var(s,\varphi) = \varnothing$,
we know $x \notin \Var(\rho)$, and thus $x \notin \Var(r)$.
We claim that $x \in \Var(s)$ and assume the contrary.
Then by $x \in \Var(t)$, we know $x \in \Var(r\gamma) \setminus \Var(\ell\gamma)$.
This implies by $x \notin \Var(r)$
that for some $x' \in \Dom(\gamma)$ we have $x \in \Var(x'\gamma)$.
By $\Dom(\gamma)= \Var(\ell)$, this contradicts $x \notin \Var(\ell\gamma)$.
Thus, $x \in \Var(s)$.
From $x \in \Var(s)$ and $x \in \Var(\varphi)$,
$x \in \Var(s) \cap \Var(\varphi) = \FVar(\ECO{\vec{x}}{\varphi})
\subseteq X$.
Thus, $x \in X \cap \Var(t)$.

\item 
Case $x \in \Var(\pi\gamma) \cap \Var(t)$.
Suppose there exists some $x' \in \Dom(\gamma) \cap \Var(\pi)$ such that
$x \in \Var(x'\gamma)$. Then, $x' \in \Var(\ell)$,
and thus $x \in \Var(\ell\gamma)$.
Since $x' \in \Var(\pi) \subseteq Z$, we know $x' \in \Var(\ell) \cap Z$.
This implies $x = \gamma(x') \in \Val \cup X$ by condition (3) above.
Hence $x \in X$. Thus, $x \in X \cap \Var(t)$.
Now, suppose he opposite, i.e.\ 
there is no $x' \in \Dom(\gamma) \cap \Var(\pi)$ such that $x \in \Var(x'\gamma)$.
Then from $x \in \Var(\pi\gamma)$, we know $x \in \Var(\pi) \setminus \Var(\ell)$.
Thus, $x \notin \Var(s,\varphi)$ by $\Var(\rho) \cap \Var(s,\varphi) = \varnothing$.
Since $x \in \Var(t)$, this implies $x \in \ExVar(\rho)$.
\end{itemize}
\end{proof}

\LemmaIVviii*

\begin{proof}
Let $\CTerm{X}{s}{\vec{x}}{\varphi} \R^p_{\rho,\gamma} \CTerm{Y}{t}{\vec{y}}{\psi}$.
Then we have 
(1) $\Dom(\gamma) = \Var(\ell)$,
(2) $s|_p = \ell\gamma$,
(3) $\gamma(x) \in \Val \cup X$ for any $x \in \Var(\ell) \cap Z$, and
(4) $\vDash_\xM (\ECO{\vec{x}}{\varphi}) \Rightarrow (\ECO{\vec{z}}{\pi\gamma})$,
where $\SET{\vec{z}} = \Var(\pi) \setminus \Var(\ell)$,
$t = s[r\gamma]$,
$\psi = \varphi \land \pi\gamma$,
$\SET{\vec{y}} = \Var(\psi) \setminus \Var(t)$, and
$Y = \ExVar(\rho) \cup (X \cap \Var(t))$.

\textup{\Bfnum{1.}}
We have $\BVar(\ECO{\vec{x}}{\varphi}) 
= \Var(\varphi) \setminus \Var(s)
= \Var(\varphi) \setminus \Var(s[\ell\gamma]_{p})$
and 
$\BVar(\ECO{\vec{y}}{\psi}) 
= \Var(\psi) \setminus \Var(t)
= \Var(\varphi \land \pi\gamma) \setminus \Var(s[r\gamma]_{p})$.
Thus, assume $x \in \Var(\varphi) \setminus \Var(s[\ell\gamma]_p)$,
i.e.\ $x \in \Var(\varphi)$ and $x \notin \Var(s[\ell\gamma]_p)$.
We show $x \in \Var(\varphi \land \pi\gamma) \setminus \Var(s[r\gamma])$.
By $x \in \Var(\varphi)$, clearly $x \in \Var(\varphi \land \pi\gamma)$.
It remains to show $x \notin \Var(s[r\gamma]_p)$.
Assume the contrary, i.e.\ $x \in \Var(s[r\gamma]_p)$.
Then by $x \notin \Var(s[\ell\gamma]_p)$ we have that $x \notin \Var(s[~]_p)$.
Thus, $x \in \Var(r\gamma)$.
If $x \in \Var(x'\gamma)$ for some $x' \in \Dom(\gamma) = \Var(\ell)$,
then $x = \gamma(x') \in \Var(\ell\gamma)$, 
which contradicts our assumption $x \notin \Var(s[\ell\gamma]_p)$.
Thus, there is no $x' \in \Dom(\gamma) = \Var(\ell)$ such that $x \in \Var(x'\gamma)$.
This implies that $x \in \Var(r) \setminus \Var(\ell)$.
However, since $\Var(\rho) \cap \Var(s,\varphi) = \varnothing$,
this contradicts our assumption $x \in \Var(\varphi)$.
Therefore we conclude that $x \notin \Var(s[r\gamma]_p)$.

\textup{\Bfnum{2.}}
%%$\BVar(\ECO{\vec{x}}{\varphi}) \cap \Var(r\gamma, \pi\gamma) = \varnothing$.
Firstly,
$\SET{\vec{x}} =  \Var(\varphi) \setminus \Var(s)$
and hence $\SET{\vec{x}} \cap \Var(s) = \varnothing$.
We also have $s|_p = \ell\gamma$ and then that $\Var(l\gamma) \subseteq \Var(s)$.
This implies $\SET{\vec{x}} \cap \Var(\ell\gamma) = \varnothing$.
Secondly, from~\Bfnum{1.} we obtain that $\SET{\vec{x}} \subseteq \SET{\vec{y}}$.
Since $\SET{\vec{y}} = \Var(\psi) \setminus \Var(t)$ we have that
$\SET{\vec{y}} \cap \Var(t) = \SET{\vec{y}} \cap \Var(s[r\gamma]_p) = \varnothing$.
Hence $\SET{\vec{y}} \cap \Var(r\gamma) = \varnothing$.

It suffices to show that 
$\BVar(\ECO{\vec{x}}{\varphi}) \cap \Var(\pi\gamma) = \varnothing$.
Suppose the contrary, for which there exists a variable $x$ 
such that $x \in \Var(\varphi)$, $x \notin \Var(s)$ and $x \in \Var(\pi\gamma)$.
Based on $x \in \Var(\pi\gamma)$ we perform a case analysis:
\begin{itemize}
\item Assume that there exists an $x' \in \Dom(\gamma) \cap \Var(\pi) = \Var(\ell) \cap \Var(\pi)$
such that $x \in \Var(x'\gamma)$.
Then we have $x \in \Var(\ell\gamma) \subseteq \Var(s)$, which contradicts our assumption.
\item Now assume that there is no $x' \in \Dom(\gamma) \cap \Var(\pi) = \Var(\ell) \cap \Var(\pi)$
such that $x \in \Var(x'\gamma)$.
By $x \in \Var(\pi\gamma)$ it follows that
$x \in \Var(\pi) \setminus \Dom(\gamma) = \Var(\pi) \setminus \Var(\ell)$.
Then $x \in \Var(\rho)$, but as $\Var(\rho)\cap \Var(s,\varphi) = \varnothing$
it contradicts that $x \in \Var(\varphi)$.
\end{itemize}

\textup{\Bfnum{3.}}
($\subseteq$)
From $Y = \ExVar(\rho) \cup (X \cap \Var(t))$,
the case $x \in Y$ is clear.
The case $x \in \BVar(\ECO{\vec{y}}{\psi})$
is also obvious, as 
$\BVar(\ECO{\vec{y}}{\psi}) \subseteq \Var(\psi)$.
($\supseteq$)
Due to the fact that $Y = \ExVar(\rho) \cup (X \cap \Var(t))$,
it suffices to show $\Var(\psi) \subseteq Y \cup \BVar(\ECO{\vec{y}}{\psi})$.
Furthermore, as $\Var(\psi) = \FVar(\ECO{\vec{y}}{\psi}) \cup \BVar(\ECO{\vec{y}}{\psi})$
it only remains to show that 
$\FVar(\psi) \subseteq Y = \ExVar(\rho) \cup (X \cap \Var(t))$.
The latter is satisfied by \Cref{lem:free variables of constraint of reducts}.
\end{proof}

\begin{restatable}{lemma}{LemmaIVxviii}
\label{lem:bullet translation and contexts}
Let $X$ be a set of variables and $s$ a term.
For any subterm $f(u_1,\ldots,u_k)$ of $s$ such that $f \notin \Val$,
we have $f(u_1,\ldots,u_k)^{\bullet_X} = f(u_1^{\bullet_X},\ldots,u_k^{\bullet_X})$.
\end{restatable}

\begin{proof}
Let $X$ be a set of variables and $s$ a term.
Suppose that $\Pos_{X \cup \Val}(s)= \SET{\seq{p}}$
and $s^{\bullet_X} = s[\seq{x}]_{\seq{p}}$.
Let $s|_p = f(u_1,\ldots,u_k)$ with $f \not\in \Val$ and $p \in \Pos(s)$. 
Then $s(p) = f$ and $s|_{p \cdot j} = u_j$ for each $1 \le j \le k$.
Since $s(p) \notin \Val$, we know that $p \notin \SET{\seq{p}}$,
and thus, $s^\bullet(p) = s[\seq{x}]_{\seq{p}}(p) = s(p) = f$.
Moreover, $u_j^\bullet = (s|_{p \cdot j})^\bullet = s^\bullet|_{p \cdot j}$ for each $1 \le j \le k$.
Thus, $f(u_1^\bullet,\ldots,u_k^\bullet)
= f(s^\bullet|_{p \cdot 1},\ldots,s^\bullet|_{p \cdot k})
= s^\bullet(p)(s^\bullet|_{p \cdot 1},\ldots,s^\bullet|_{p \cdot k})
= s^\bullet|_p$.
On the other hand, we have
$f(u_1,\ldots,u_k)^\bullet = (s|_p)^\bullet = s^\bullet|_p$ by definition.
This completes the proof.
\end{proof}

\LemmaIVxix*
\begin{proof}
First note that $s^\bullet\sigma = s$ and
suppose $u = s|_p$.
Then $u^\bullet\sigma = (s^\bullet|_{p})\sigma = (s^\bullet\sigma)|_{p} = s|_p  = u$.
\end{proof}

\LemmaIVxx*

\begin{proof}
Let $X$ be a set of variables,
$s$ a term, and $\ell$ a linear term such that $s|_p = \ell\gamma$ for a
substitution $\gamma$ with $\Dom(\gamma) = \Var(\ell)$. Assume further that
$\ell$ is value-free. We prove by induction on the subterms $u$ of $\ell$ that
$(u\gamma)^{\bullet_{X}} = u\gamma^{\bullet_{X}}$ holds.
\begin{itemize}
\item Assume $u = x \in \xV$.
Then $(x\gamma)^\bullet = \gamma(x)^\bullet 
= \gamma^\bullet(x) 
= x\gamma^\bullet$ by definition of $\gamma^\bullet$.
\item Assume $u = f(u_1,\ldots,u_n)$. By our assumption we have
that $f \notin \Val$ and $u_1,\ldots,u_n$ are value-free. Thus, using 
\Cref{lem:bullet translation and contexts} and
the
induction hypothesis, we have that
\[
\begin{array}{@{}l@{\>}c@{\>}l@{}}
(u\gamma)^\bullet 
&=& (f(u_1\gamma,\ldots,u_n\gamma))^\bullet\\
&=& f((u_1\gamma)^\bullet,\ldots,(u_n\gamma)^\bullet)\\
&=& f(u_1\gamma^\bullet,\ldots,u_n\gamma^\bullet)\\
&=& f(u_1,\ldots,u_n)\gamma^\bullet\\
&=& u\gamma^\bullet
\end{array}
\]
\end{itemize}
Note that if $u = v_i \in \Val$ then $(u\gamma)^\bullet = (v_i)^\bullet =
x_i \neq v_i = u\gamma^\bullet$ for a fresh variable $x_i$.
Thus, the condition of value-freeness of
$\ell$ is essential here. Finally, we conclude $s^\bullet|_p =
s|_p^\bullet = (\ell\gamma)^\bullet = \ell\gamma^\bullet$.
\end{proof}

\LemmaIVxxi*

\begin{proof}
Let $\rho$ be a valuation and suppose that 
$\vDash_{\xM,\rho} \bigwedge_{i=1}^n (s(p_i) = x_i)$.
We show that $\vDash_{\xM,\rho} (u = u^{\bullet_X})$ holds for any subterm $u$ of $s$
by induction on $u$.

\begin{itemize}
\item 
Assume $u = s(p_i)$.
Then, we have $u^\bullet = s(p_i)^\bullet  = s^\bullet|_{p_i}  
= (s[x_1,\ldots,x_n]_{p_1,\ldots,p_n})|_{p_i}  
= x_i$.
Since $\vDash_{\xM,\rho} (s(p_i) = x_i)$ follows
from $\vDash_{\xM,\rho} \bigwedge_{i=1}^n (s(p_i) = x_i)$,
we have $\vDash_{\xM,\rho} (u = u^\bullet)$.

\item 
Assume $u = x \in \Var(s) \setminus X$.
Suppose $x = s(p)$ with $p \notin \SET{\seq{p}}$.
Then, we have $u^\bullet = s(p)^\bullet  = s^\bullet|_p
= (s[x_1,\ldots,x_n]_{p_1,\ldots,p_n})|_{p}
= s(p) = x$.
Clearly, we have $\vDash_{\xM,\rho} x = x$,
and hence $\vDash_{\xM,\rho} u = u^\bullet$.

\item 
Otherwise, we have $u = f(u_1,\ldots,u_k)$ with $f \notin \Val$.
Using the induction hypothesis, we obtain
$\vDash_{\xM,\rho} (u_j = u_j^\bullet)$ for $1 \le j \le k$.
Thus, 
$\vDash_{\xM,\rho} f(u_1,\ldots,u_k) = f(u_1^\bullet,\ldots,u_k^\bullet)$.
We conclude
$\vDash_{\xM,\rho} (f(u_1,\ldots,u_k) = f(u_1,\ldots,u_k)^\bullet)$
by \Cref{lem:bullet translation and contexts}.
\end{itemize}

From all of this we have
$\vDash_{\xM} (\bigwedge_{i=1}^n (s(p_i) = x_i)) \Rightarrow (u = u^{\bullet_X})$
for all subterms $u$ of $s$.
\end{proof}

\LemmaIVxii*

\begin{proof}
Assume %that 
(1)--(6) hold %holds
for two satisfiable existentially constrained terms
$\CTerm{X}{s}{\vec{x}}{\varphi}$ and $\CTerm{Y}{t}{\vec{y}}{\psi}$.
Let $\gamma$ be a $Y$-valued substitution
such that $\vDash_{\xM} (\ECO{\vec{y}}{\psi})\gamma$.
Note 
$\gamma(\FVar(\ECO{\vec{y}}{\psi})) \subseteq \Val$ holds.
By our assumption~(6),
$\vDash_{\xM} (\ECO{\vec{x}}{\varphi})\sigma\gamma$.
Define a substitution $\xi = \gamma \circ \sigma$ 
for which we have 
$\vDash_{\xM} (\ECO{\vec{x}}{\varphi})\xi$.
From assumption~(2) it follows
that $s\xi = s\sigma\gamma = t\gamma$.
It remains to show that $\xi(X) \subseteq \Val$.
Fix an arbitrary $x \in X$. If $\sigma(x) \in \Val$ then 
$\xi(x) = x\sigma\gamma \in \Val$, and we are done.
Otherwise, we have $\sigma(x) \in \xV$ by the assumption~(1) on $\sigma$.
Hence, by the assumption~(5) follows that $\sigma(x)\in Y$.
As $\gamma$ is $Y$-valued,
we have $\xi(x) = x\sigma\gamma \in \Val$.

For the other direction, suppose an $X$-valued substitution $\xi$ such that %%$\xi(X) \subseteq \Val$ and 
$\vDash_{\xM} (\ECO{\vec{x}}{\varphi})\xi$.
Note here that $\xi(\FVar(\ECO{\vec{x}}{\varphi})) \subseteq \Val$.
Then there exists $\kappa = \SET{\vec{x} \mapsto \vec{v}}$ 
for $\SET{\vec{v}} \subseteq \Val$
such that $\vDash_{\xM} \varphi\kappa\xi$
with $\Var(\varphi\kappa\xi) = \varnothing$.
We show now for any $x,x' \in V$
if $\sigma(x) = \sigma(x')$ then $x\xi = x'\xi$.
Suppose $\sigma(x) = \sigma(x')$.
Then by our assumption~(3),
$\vDash_\xM (\ECO{\vec{x}}{\varphi}) \Rightarrow x = x'$.
Hence, $\vDash_\xM (\ECO{\vec{x}}{\varphi})\xi \Rightarrow x\xi = x'\xi$.
Thus, by $\vDash_{\xM} (\ECO{\vec{x}}{\varphi})\xi$,
we obtain $\vDash_\xM x\xi = x'\xi$.
As $x,x' \in V \subseteq X$ and $\xi$ is $X$-valued,
$x\xi,x'\xi \in \Val$.
This gives in the end that $x\xi = x'\xi$.

Again assume that $x,x' \in V$.
Define a substitution $\delta$
by $\delta(y) = x\xi$ if $y \in \sigma(V)$ and $\sigma(x) = y$,
and $\delta(y) = \xi(y)$ otherwise.
Indeed, $\delta$ is well-defined,
as if $\sigma(x) = \sigma(x') = y$ then 
$x\xi = x'\xi$ as shown above.
By definition, we have $x\sigma\delta = x\xi$ for any $x \in V$.
Also, for $x \in \Var(s) \setminus V$,
we have $x\sigma = x$ by our condition (1)
and $\delta(x) = \xi(x)$ by definition of $\delta$.
Thus, it follows $x\sigma\delta = x\delta = x\xi$.
Hence, we conclude that $x\sigma\delta = x\xi$ for any $x \in \Var(s)$.
By condition~(2) %%$s\sigma = t$,
it follows that $s\xi = s\sigma\delta = t\delta$.

Let $y \in Y$. We show that $y\delta \in \Val$.
Suppose $y \in \sigma(V)$, that is, $\sigma(x) = y$ for some $x \in V$.
Then, $y\delta = x\xi$.
Thus, as $\xi$ is $X$-valued, it follows that $\delta(y) = \xi(x) \in \Val$.
Otherwise, we have that $y \in Y \setminus \sigma(V)$.
Then by our condition (4), $y \in X \setminus V$, and hence $y \in X$.
Then we have that $\delta(y) = \xi(y)$ by definition of $\delta$, as $y \notin \sigma(V)$.
Then, as $\xi$ is $X$-valued, it follows that $\delta(y) = \xi(y) \in \Val$.
Hence $\delta$ is $Y$-valued.

Since $\vDash_\xM x\xi = x\sigma\delta$ for all $x \in X$,
by $\FVar(\ECO{\vec{x}}{\varphi}) \subseteq X$,
we have that $\vDash_\xM x\xi = x\sigma\delta$ for any $x \in \FVar(\ECO{\vec{x}}{\varphi})$.
Thus, from $\vDash_\xM (\ECO{\vec{x}}{\varphi})\xi$,
we obtain $\vDash_\xM (\ECO{\vec{x}}{\varphi})\sigma\delta$ and hence
$\vDash_\xM (\ECO{\vec{y}}{\psi})\delta$ by our condition~(6).
We conclude 
$\CTerm{X}{s}{\vec{x}}{\varphi} \sim \CTerm{Y}{t}{\vec{y}}{\psi}$.
\end{proof}

\TheoremIVxxv*

\begin{proof}
We show by induction on $n$
that if
$\CTerm{X}{s}{\vec{x}}{\varphi}
(\sim  \cdot \to_\xR)^n \cdot \sim \CTerm{Y}{t}{\vec{y}}{\psi}$,
then there exists an existentially constrained term $\CTerm{X'}{s'}{\pvec{x}'}{\varphi'}$
such that 
$\CTerm{X}{s}{\vec{x}}{\varphi} \to^*_\xR \CTerm{X'}{s'}{\pvec{x}'}{\varphi'} 
\sim \CTerm{Y}{t}{\vec{y}}{\psi}$.
The case $n = 0$ is trivial.
Otherwise $n > 0$, thus we have 
$\CTerm{X}{s}{\vec{x}}{\varphi} 
(\sim  \cdot \to_\xR)^{n-1} \cdot \sim  
\CTerm{Z}{u}{\vec{z}}{\delta}
\to_\xR 
\CTerm{Z'}{u'}{\pvec{z}'}{\delta'}
\sim \CTerm{Y}{t}{\vec{y}}{\psi}$
for some 
$\CTerm{Z}{u}{\vec{z}}{\delta}$ and $\CTerm{Z'}{u'}{\pvec{z}'}{\delta'}$.
Applying the induction hypothesis
to 
$\CTerm{X}{s}{\vec{x}}{\varphi} 
(\sim  \cdot \to_\xR)^{n-1}  \sim  
\CTerm{Z}{u}{\vec{z}}{\delta}$
yields an existentially constrained term 
$\CTerm{X''}{s''}{\pvec{x}''}{\varphi''}$
such that
$\CTerm{X}{s}{\vec{x}}{\varphi} 
\to^*_\xR 
\CTerm{X''}{s''}{\pvec{x}''}{\varphi''}
\sim  \CTerm{Z}{u}{\vec{z}}{\delta}$.
Thus,
$\CTerm{X''}{s''}{\pvec{x}''}{\varphi''}
\sim  \CTerm{Z}{u}{\vec{z}}{\delta}
\to_\xR 
\CTerm{Z'}{u'}{\pvec{z}'}{\delta'}$.
By \Cref{thm:commutativity of rewrite steps and equivalence by left-value-free rules},
there exists an existentially constrained term $\CTerm{X'}{s'}{\pvec{x}'}{\varphi'}$ such that we have
$\CTerm{X''}{s''}{\pvec{x}''}{\varphi''}
\to_\xR 
\CTerm{X'}{s'}{\pvec{x}'}{\varphi'}
\sim 
\CTerm{Z'}{u'}{\pvec{z}'}{\delta'}$.
All in all, we have that
$\CTerm{X}{s}{\vec{x}}{\varphi} 
\to^*_\xR 
\CTerm{X''}{s''}{\pvec{x}''}{\varphi''}
\to_\xR 
\CTerm{X'}{s'}{\pvec{x}'}{\varphi'}
\sim 
\CTerm{Z'}{u'}{\pvec{z}'}{\delta'}
\sim \CTerm{Y}{t}{\vec{y}}{\psi}$.
Hence the claim holds.
\end{proof}

\end{document}